\documentclass{article}

\PassOptionsToPackage{numbers}{natbib}
\usepackage[final]{neurips_2020}
\usepackage[utf8]{inputenc}

\usepackage{booktabs} 
\usepackage[ruled]{algorithm2e} 

\SetAlFnt{\small}
\SetAlCapFnt{\small}
\SetAlCapNameFnt{\small}
\SetAlCapHSkip{0pt}
\IncMargin{-\parindent}

\usepackage{graphicx}
\usepackage{amsmath,amssymb,amsthm}
\usepackage{bm}
\usepackage{rotating}
\usepackage{multicol}
\usepackage{titlesec}
\usepackage{latexsym}
\usepackage{color}
\usepackage{hyperref}

\newcounter{newct}

\newcommand{\bv}{\begin{array}}

\newcommand{\x}{{ \bf x}}

\newcommand{\appProp}[2]{{\bf Proposition~\ref{#1}.} {\em #2}}

\newcommand{\appLem}[2]{{\bf Lemma~\ref{#1}.} {\em #2}}
\newcommand{\appThm}[2]{{\bf Theorem~\ref{#1}.} {\em #2}}

\newtheorem{thm}{Theorem}
\newtheorem{dfn}{Definition}

\newtheorem{lem}{Lemma}
\newtheorem{ex}{Example}

\newtheorem{prop}{Proposition}
\newtheorem{coro}{Corrollary}

\newtheorem{claim}{Claim}

\newcommand{\ma}{\mathcal A}
\newcommand{\mm}{\mathcal M}
\newcommand{\ms}{\mathcal S}

\newcommand{\ml}{\mathcal L}

\renewcommand{\cal}{\mathcal }

\newcommand{\calT}{\mathcal T}
\newcommand{\calQ}{\mathcal Q}
\newcommand{\ra}{\rightarrow}

\newcommand{\cw}{\rel_\text{CW}}
\newcommand{\mpsr}{\text{MPSR}}
\newcommand{\cor}{c} 
\newcommand{\rank}{\text{Rank}}
\newcommand{\tb}{\text{TB}}

\newcommand{\kt}{\text{KT}}
\newcommand{\mallows}{\text{Ma}}
\newcommand{\pl}{\text{Pl}}

\newcommand{\expect}{{\mathbb E}}
\newcommand{\perm}{\text{Perm}}
\newcommand{\hist}{\text{Hist}}

\newcommand{\ind}{\text{Ind}}

\newcommand{\rel}{\text{S}}
\newcommand{\id}{\text{Id}}
\newcommand{\gm}{{\cal G}_m}

\newcommand{\Omit}[1]{}
\newcommand{\calU}{\mathcal U}
\newcommand{\mc}{\mathcal C}
\newcommand{\mg}{\mathcal G}
\newcommand{\mx}{\mathcal X}

\newcommand{\ba}{{\mathbf E}}
\newcommand{\bb}{{\mathbf S}}
\newcommand{\pba}[1]{{\mathbf E}^{#1}}
\newcommand{\pbb}[1]{{\mathbf S}^{#1}}

\newcommand{\bd}{{\mathbf D}}
\newcommand{\cons}{\text{C}^{\ba\bb}}
\newcommand{\pcons}[1]{\text{C}^{#1}}
\newcommand{\rcons}{\text{C}^{\ba\bb}_{\le 0}}
\newcommand{\consa}{\text{C}^{\ba}}
\newcommand{\consap}{\text{C}^{\ba'}}
\newcommand{\conv}{\text{CH}}

\newcommand{\ppoly}[1]{{\mathcal H}^{#1}}
\newcommand{\ppolyz}[1]{{\mathcal H}_{\le 0}^{#1}}
\newcommand{\sol}{{\mathcal H}}
\newcommand{\rsol}{{\mathcal H}_{\le 0}}
\newcommand{\wmg}{\text{WMG}}
\newcommand{\umg}{\text{UMG}}

\newcommand{\pair}{\text{Pair}}
\newcommand{\wcw}{\text{WCW}}
\newcommand{\ncc}{\rel_\text{NCC}}

\newcommand{\ties}{\text{Ties}}
\newcommand{\sgroup}{{\cal S}_{\ma}}
\newcommand{\lex}{{\sc Lex}}
\newcommand{\fa}{{\sc Fa}}
\newcommand{\rt}{\text{RunTime}}
\newcommand{\loss}{\text{Loss}}

\newcommand{\invert}[1]{(#1)^\top}

\begin{document}

\title{The Smoothed Possibility of Social Choice
}
\author{Lirong Xia, RPI, xialirong@gmail.com
}\date{}
\maketitle
\begin{abstract} 
We develop a framework that leverages the smoothed complexity analysis by~\citet{Spielman2004:Smoothed} to circumvent paradoxes and impossibility theorems in social choice, motivated by modern applications of social choice powered by AI and ML. 
For Condrocet's paradox, we prove that the smoothed likelihood of  the paradox either vanishes at an exponential rate as the number of agents increases, or does not vanish at all. For the ANR impossibility on the non-existence of voting rules that simultaneously satisfy anonymity, neutrality, and resolvability, we characterize the rate for the impossibility to vanish, to be either polynomially fast or exponentially fast. We also propose a novel easy-to-compute tie-breaking mechanism that optimally preserves anonymity and neutrality for even number of alternatives in natural settings. Our results illustrate the smoothed possibility of social choice---even though the paradox and the impossibility theorem hold in the worst case, they may not be a big concern in practice. 
\end{abstract}

\section{Introduction}
Dealing with paradoxes and impossibility theorems is a major challenge in social choice theory, because ``{\em the force and widespread presence of impossibility results generated a consolidated sense of pessimism, and this became a dominant theme in welfare economics and social choice theory in general}'', as the eminent economist Amartya Sen commented in his Nobel prize lecture~\cite{Sen1999:The-Possibility}.

Many paradoxes and impossibility theorems in social choice are based on worst-case analysis. Take perhaps the earliest one, namely {\em Condorcet's (voting) paradox}~\cite{Condorcet1785:Essai}, for example. Condorcet's paradox states that, when there are at least three alternatives, it is impossible for pairwise majority aggregation to be transitive. 
The proof  is done by explicitly constructing a worst-case scenario---a profile $P$ that contains a {\em Condorcet cycle}. For example, in $P=\{a\succ b\succ c, b\succ c\succ a, c\succ a\succ b\}$, there is a cycle $a\succ b$, $b\succ c$, and $c\succ a$ of pairwise majority. Condorcet's paradox is closely related to the celebrated Arrow's impossibility theorem: if Condorcet's paradox can be avoided, then  the pairwise majority rule can avoid Arrow's impossibility theorem.

As another example, the {\em ANR impossibility theorem}  (e.g.~\cite[Problem~1]{Moulin1983:The-Strategy} and~\cite{Ozkes2020:Anonymous,Dogan2015:Anonymous,Campbell2015:The-finer})  states that no voting rule $r$ can simultaneously satisfy anonymity ($r$ is insensitive to the identities of agents) and neutrality ($r$ is insensitive to the identities of alternatives), and resolvability ($r$ always chooses a single winner). The proof is done by  analyzing a worst-case scenario $P=\{a\succ b, b\succ a\}$. Suppose  for the sake of contradiction that a resolvable $r$ satisfies anonymity and neutrality, and without loss of generality let $r(P)=a$. After exchanging $a$ and $b$, the winner ought to be $b$ due to neutrality. But since the permuted profile still contains one vote for $a\succ b$ and one vote for $b\succ a$, the winner ought to be $a$ due to anonymity, which is a contradiction.

There is an enormous literature in social choice on circumventing the impossibilities, most of which belongs to the following two approaches. {\em (1) Domain restrictions}, namely, agents' reported preferences are assumed to come from a subset of all linear orders such as single-peaked preferences~\cite{Black58:Theory,Arrow63:Social,Sen1966:A-Possibility,Moulin80:Strategy,Conitzer09:Eliciting,Brandt10:Bypassing,Elkind2016:Preference}; and {\em (2) likelihood analysis}, where impossibility theorems are evaluated by the likelihood of their occurrence in profiles randomly generated from a distribution such as the i.i.d.~uniform distribution, a.k.a.~{\em Impartial Culture (IC)}~\cite{Gehrlein2002:Condorcets,Green-Armytage2016:Statistical}. Both approaches have been criticized for making strong and unrealistic assumptions on the domain and on the probability distributions, respectively. In particular, IC has received much criticism, see e.g.~\cite{Deemen2014:On-the-empirical}, yet no widely-accepted probabilisitic model exists to the best of our knowledge. 


The worst-case nature behind the impossibility theorems might be desirable for high-stakes, less-frequent applications such as political elections, but it may not be appropriate for modern low-stakes, frequently-used applications of social choice, many of which are supported by AI systems that   learn agents' preferences to help them make group decisions~\cite{Xia2017:Improving}. While AI-powered social choice appears to be a promising solution to the long-standing turnout problem~\citep{Sen1999:The-Possibility} and can therefore promote democracy to a larger scale and with a higher frequency, it questions the relevance of worst-case analysis in social choice theory. 
This motivates us to ask the following key question:

{\bf How serious are the impossibilities in  frequently-used modern applications of social choice?}

The  frequently-used feature naturally leads to the analysis of average likelihood of the impossibility theorems. But in light of the criticism of the classical likelihood analysis approach discussed above, what is a realistic model to answer the key question?

Interestingly, computer science has encountered a similar challenge and has gone through a similar path in the analysis of practical performance of algorithms. Initially, the analysis mostly focused on the worst case, as in the spirit of big $O$ notation and {\sc NP}-hardness. Domain restrictions have  been a popular approach beyond the worst-case analysis. For example, while SAT is NP-hard, its restriction 2-SAT is in P. 
Likelihood analysis, in particular average-case complexity analysis~\cite{Bogdanov2006:Average-Case}, has also been a popular approach, yet it suffers from the same criticism as its counterpart in social choice---the distribution used in the analysis may well be unrealistic~\cite{Spielman2009:Smoothed}.

The challenge was addressed by the {\em smoothed (complexity) analysis} introduced by~\citet{Spielman2004:Smoothed}, which focuses on the ``worst average-case'' scenario that combines the worst-case analysis and the average-case analysis. The idea is based on the fact that the  input $\vec x$ of an algorithm is often a noisy perception of the ground truth $\vec x^*$. Therefore, the worst-case is analyzed by assuming that an adversary chooses a ground truth $\vec x^*$ and then Nature adds a noise $\vec\epsilon$ (e.g.~a Gaussian noise) to it, such that the algorithm's input becomes $\vec x=\vec x^*+\vec \epsilon$. The smoothed runtime of an algorithm is therefore $\sup_{\vec x^*}\expect_{\vec\epsilon}\  \rt(\vec x^*+\vec\epsilon)$, in contrast to the worst-case runtime $\sup_{\vec x^*}\rt(\vec x^*)$ and the average-case runtime $\expect_{\vec x^*\sim\pi} \rt(\vec x^*)$, where $\pi$ is a given distribution over data.

\noindent{\bf  Our Contributions.} We propose a framework that leverages the elegant smoothed complexity analysis to answer the key question above. In social choice, the data is a {\em profile}, which consists of agents' reported preferences that are often represented by linear orders over a set $\ma$ of $m$ alternatives. Like in the smoothed complexity analysis, in our framework there is an adversary who controls agents' ``ground truth'' preferences, which may be ordinal (as rankings over alternatives) or cardinal (as utilities over alternatives). Then, Nature adds a ``noise'' to the ground truth preferences and outputs a preference profile, which consists of linear orders over alternatives. 

Following the convention in average-case complexity analysis~\cite{Bogdanov2006:Average-Case}, we use a  statistical model to model Nature's noising procedure. 
As in many smoothed-analysis approaches, we assume that noises in agents' preferences are independently generated, yet agents' ground truth preferences can be arbitrarily correlated, which constitutes the basis for the worst-case analysis. 
Using our smoothed analysis framework, we obtain the following two dichotomy theorems on the asymptotic smoothed likelihood of Condorcet's paradox and the ANR impossibility under mild assumptions, when the number of alternatives $m$ is fixed and the number of agents $n$ goes to infinity.  

\vspace{1mm}
{\noindent \sc Theorem~\ref{thm:smoothedCondorcet}.} {\bf (Smoothed Condorcet's paradox, informally put).} {\em The smoothed likelihood of Condorcet's Paradox either vanishes at an exponential rate, or does not vanish at all.
}

\vspace{1mm}
{\noindent \sc Theorem~\ref{thm:smoothedANR}.} {\bf (Smoothed ANR (im)possibility theorem, informally put).} {\em The theorem has two parts. The {\bf smoothed possibility} part states that there exist resolute voting rules under which the impossibility theorem either vanishes at an exponential rate or at a polynomial rate. The {\bf smoothed impossibility part} states that there does not exists a resolute voting rule under which the impossibility theorem vanishes   faster than under the rules in the smoothed possibility part.
}

Both theorems are quite general and 
their formal statements also characterize conditions for each case. Such conditions in Theorem~\ref{thm:smoothedCondorcet} tell us when  Condorcet's Paradox vanishes (at an exponential rate), which is positive. While the theorem may be expected at a high level and part of it is easy to prove, for example the exponential-rate part can be proved by using a similar idea as in the proof of minimaxity/sample complexity of MLE under a large class of distance-based models~\citep{Caragiannis2016:When}, we are not aware of a previous work that provides a {\em complete dichotomy} that draws a clear line between paradoxes and non-paradoxes as Theorem~\ref{thm:smoothedCondorcet} does. In addition, we view such expectedness positive news, because it provides a theoretical confirmation of well-believed hypotheses under natural settings, as smoothed complexity analysis did for the runtime of a simplex algorithm.

The smoothed possibility part of Theorem~\ref{thm:smoothedANR} is also positive because it states that the ANR impossibility vanishes as the number of agents $n$ increases. 
The smoothed impossibility part of Theorem~\ref{thm:smoothedANR} is mildly negative, because it states that no voting rule can do better, though the impossibility may still vanish as $n$ increases. Together, Theorem~\ref{thm:smoothedCondorcet}  and~\ref{thm:smoothedANR}  illustrate the smoothed possibility of social choice---even though the paradox and the impossibility theorem hold in the worst case, they may not be a big concern in practice in some natural settings. 

Our framework also allows us to develop a novel easy-to-compute tie-breaking mechanism called {\em most popular singleton ranking (\mpsr)} tie-breaking, which tries to break ties using a linear order that uniquely occurs most often in the profile (Definition~\ref{dfn:mpsrtb}). We prove that  $\mpsr$ is better than the commonly-used lexicographic tie-breaking and fixed-agent tie-breaking mechanisms w.r.t.~the smoothed likelihood of the ANR impossibility---\mpsr{} reduces the smoothed likelihood from $n^{-0.5}$ to $n^{-\frac{m!}{4}}$ for many commonly-studied voting rules under natural assumptions (Proposition~\ref{prop:nonoptlexfa} and Theorem~\ref{thm:mpsropt}), and is optimal for even number of alternatives $m$ (Theorem~\ref{thm:smoothedANR} and Lemma~\ref{lem:taubounds}).

\vspace{1mm}
\noindent{\bf Proof Techniques.} Standard approximation techniques such as Berry-Esseen theorem and its high-dimensional counterparts, e.g.~\cite{Valiant2011:Estimating,Daskalakis2016:A-Size-Free,Diakonikolas2016:The-fourier}, due to their  $O(n^{-0.5})$ error terms, are too coarse for the (tight) bound in Theorem~\ref{thm:smoothedANR}. To prove our theorems, we first model various events of interest as systems of linear constraints.  Then, we develop a technical tool (Lemma~\ref{lem:maintech}) to provide a dichotomy characterization  for the {\em Poisson Multinomial Variables (PMV)} that corresponds to the histogram of a randomly generated profile to satisfy the constraints. We further show in Appendix~\ref{sec:application} that Lemma~\ref{lem:maintech} is a general and useful tool for analyzing smoothed likelihood of many other commonly-studied events in social choice (Table~\ref{tab:propositions}), which are otherwise hard to analyze. 


\subsection{Related Work and Discussions} 
\noindent{\bf Smoothed analysis.} Smoothed analysis has been applied to a wide range of problems in mathematical programming, machine learning, numerical analysis, discrete math, combinatorial optimization, and equilibrium analysis  and price of anarchy~\cite{Chung2008:The-Price}, see~\cite{Spielman2009:Smoothed} for a survey.  In a recent position paper, \citet{Baumeister2020:Towards} proposed to conduct smoothed analysis on computational aspects of social choice and mentioned that their model can be used to analyze voting paradoxes and ties, but the paper does not contain technical results. Without knowing their work, we independently proposed and formulated the smoothed analysis framework for social choice in this paper.

\noindent{\bf The worst average-case idea.}  While our framework is inspired by the smoothed complexity analysis, the worst average-case idea is deeply rooted in (frequentist) statistics and can be viewed as a measure of robustness. Taking a statistical decision theory~\cite{Berger85:Statistical} point of view, the frequentist's loss of a decision rule $r:\text{Data}\ra \text{Decision}$ under a statistical model $\mm = (\Theta,\ms,\Pi)$ is measured by

$\hfill\sup\nolimits_{\theta\in \Theta}\expect_{P\sim \pi_\theta}(\loss(\theta, r(P))),\hfill$

where the expectation evaluates the average-case loss under the worst-case distribution $\pi_\theta\in\Pi$. There is a large literature in statistical aspects of social choice (e.g.~\cite{Condorcet1785:Essai,Caragiannis2016:When,Xia2016:Bayesian}) and preference learning (e.g.~\cite{Khetan2016:Data-driven,Nihar-B.-Shah2016:Estimation}) that study the frequentist loss w.r.t.~classical loss functions in statistics that depend on both $P$ and $\theta$, leading to consistency and minimaxity results. The idea is also closely related to the ``min of means'' criteria in decision theory~\citep{Gilboa1989:Maxmin}.
Our framework explicitly models smoothed likelihood of social choice events via loss functions that measure the dissatisfaction of axioms w.r.t.~the data (profile) and do not depend on the ``ground truth'' $\theta$. This is similar to smoothed complexity analysis, where the loss function is the runtime of an algorithm, which also only depends on the input data $P$ but not on $\theta$.


\noindent{\bf Correlations among agents' preferences.} In our model, agents' ground truth preferences can be arbitrarily correlated while the randomness  comes from independent noises. This is a standard assumption in smoothed complexity analysis as well as in relevant literatures in psychology, economics, and behavioral science, as evident in random utility models, logistic regression, MLE interpretation of the ordinary least squares method, etc.~\cite{Train09:Discrete,Xia2019:Learning}. As another justification, the adversary can be seen as a manipulator who wants to control agents' reported preferences, but is only able to do it in a probabilistic way.

\noindent{\bf Generality of results and techniques.}   Our technical results are quite general and can be immediately applied to classical likelihood analysis in social choice under i.i.d.~distribution, to answer open questions, obtain new results, and provide new insights. For example, a straightforward application of Lemma~\ref{lem:maintech} gives an asymptotic answer to an open question by~\citet{Tsetlin2003:The-impartial} (after Corollary~\ref{coro:table} in Appendix~\ref{sec:application}). As another example, we are not aware of a previous work on the asymptotic likelihood of the ANR impossibility even under IC, which is a special case of Theorem~\ref{thm:smoothedANR}. 



\noindent{\bf Other related work.}   As discussed above, there is a large literature on domain restrictions and the likelihood analysis toward circumventing impossibility theorems, see for example, the book by~\citet{Gehrlein2017:Elections} for a recent survey. In particular, there is a large literature on the likelihood of Condorcet voting paradox and the likelihood of (non)-existence of Condorcet winner under i.i.d.~distributions especially IC~\cite{DeMeyer1970:The-Probability,May1971:Some,Gehrlein1976:The-probability,Tsetlin2003:The-impartial,Jones1995:Condorcet,Green-Armytage2016:Statistical,Brandt2016:Analyzing,Brandt2019:Exploring}. The IC assumption has also been used to prove quantitative versions of other impossibility theorems in social choice such as Arrow's impossibility theorem~\cite{Kalai02:Fourier,Keller2010:On-the-probability,Mossel2012:A-quantitative} and Gibbard-Satterthwaite theorem~\cite{Friedgut2011:A-quantitative,Mossel2015:A-quantitative}, as well as in judgement aggregation~\cite{Nehama2013:Approximately,Filmus2019:AND-testing}. 
Other works have studied social choice problems when each agent's preferences are represented by a probability distribution~\cite{Bachrach10:Probabilistic,Noam-Hazon2012:-On-the-Evaluation,Procaccia16:Optimal,Zhao2018:A-Cost-Effective,Noothigattu2018:A-Voting-Based,Li2019:Minimizing}. These works focused on computing the outcome efficiently, which is quite different from our goal. 

\section{Preliminaries}
\noindent{\bf Basic Setting.} Let $\ma=[m]=\{1,\ldots,m\}$ denote the set of $m\ge 3$ {\em alternatives}. Let $\ml(\ma)$ denote the set of all linear orders (a.k.a.~rankings) over $\ma$. Let $n\in\mathbb N$ denote the number of agents. Each agent uses a linear order to represent his or her preferences. The vector of $n\in\mathbb N$ agents' votes $P$ is called a {\em (preference) profile}, or sometimes an $n$-profile. For any profile $P$, let $\hist(P)\in {\mathbb Z}_{\ge 0}^{m!}$ denote the anonymized profile of $P$,  also called the {\em histogram} of $P$, which counts the multiplicity of each linear order in $P$. 
A  {\em resolute voting rule} $r$ is a
mapping from each profile to a single winner in $\ma$.  A {\em voting correspondence} $\cor$ is a mapping from  each  profile to a non-empty set of co-winners. 

\noindent{\bf Tie-Breaking Mechanisms.}  Many commonly-studied voting rules are defined as correspondences combined with a tie-breaking mechanism. For example, a  {\em positional scoring correspondence}  is characterized by a scoring vector $\vec s=(s_1,\ldots,s_m)$ with $s_1\ge s_2\ge \cdots\ge s_m$ and $s_1>s_m$. For any alternative $a$ and any linear order $R\in\ml(\ma)$, we let $\vec s(R,a)=s_i$, where $i$ is the rank of $a$ in $R$. Given a profile $P$,  the  positional scoring correspondence $c_{\vec s}$ chooses all alternatives $a$ with maximum $\sum_{R\in P}\vec s(R,a)$. For example, {\em Plurality} uses the scoring vector $(1,0,\ldots,0)$ and {\em Borda} uses the scoring vector $(m-1,m-2,\ldots,0)$. The positional scoring rule $r_{\vec s}$ chooses a single alternative by further applying a tie-breaking mechanism.  The {\em lexicographic tie-breaking}, denoted by \lex-$R$ where $R\in\ml(\ma)$, breaks ties in favor of alternatives ranked higher in $R$. The fixed-agent tie-breaking, denoted by \fa-$j$ where $1\le j\le n$, uses  agent $j$'s preferences to break ties. 

\noindent{\bf (Un)weighted Majority Graphs.}  For any profile $P$ and any pair of alternatives $a,b$, let $ P[a\succ b]$ denote the number of rankings in $P$ where $a$ is preferred to $b$. Let $\wmg(P)$ denote the {\em weighted majority graph} of $P$, which is a complete graph where the vertices are $\ma$ and the edge weights are $w_P(a,b) = P[a\succ b] - P[b\succ a]$. For any distribution $\pi$ over $\ml(\ma)$, let $\wmg(\pi)$ denote the weighted majority graph where $\pi$ is treated as a {\em fractional} profile, where for each $R\in\ml(\ma)$ there are $\pi(R)$ copies of $R$. The {\em unweighted majority graph (UMG)} of a profile $P$, denoted by $\umg(P)$, is the unweighted directed graph where the vertices are the alternatives and there is an edge $a\ra b$ if and only if $P[a\succ b] > P[b\succ a]$. If $a$ and $b$ are tied, then there is no edge between $a$ and $b$. $\umg(\pi)$ is defined similarly. A {\em Condorcet cycle}  of a profile $P$ is a cycle in $\umg(P)$. A {\em weak Condorcet cycle} of a profile $P$ is a cycle in any supergraph of $\umg(P)$.



\noindent{\bf Axiomatic Properties.}  A voting rule $r$ satisfies {\em anonymity}, if the winner is insensitive to the identity of the voters. That is, for any pair of profiles $P$ and $P'$ with $\hist(P)=\hist(P')$, we have $r(P)=r(P')$. $r$ satisfies {\em neutrality} if the winner is insensitive to the identity of the alternatives. That is, for any permutation $\sigma$ over $\ma$, we have $r(\sigma(P))=\sigma(r(P))$, where $\sigma(P)$ is the obtained from $P$ by permuting alternatives according to $\sigma$. 

\noindent{\bf Single-Agent Preference Models.}  A statistical model $\mm=(\Theta, \ms, \Pi)$ has three components: the parameter space $\Theta$, which contains the ``ground truth''; the sample space $\ms$, which contains all possible data; and the set of probability distributions $\Pi$, which contains a distribution $\pi_{\theta}$ over $\ms$ for each $\theta\in\Theta$. In this paper we adopt single-agent preference models, where $\ms=\ml(\ma)$. 

\begin{dfn}A {\em single-agent preference model} is denoted by $\mm=(\Theta,\ml(\ma),\Pi)$. $\mm$ is {\em strictly positive} if there exists $\epsilon>0$ such that the probability of any ranking under any distribution in $\Pi$ is at least $\epsilon$.  $\mm$ is {\em closed} if $\Pi$ is a closed set in ${\mathbb R}_{\ge 0}^{m!}$, where each distribution in $\Pi$ is viewed as a vector in $m!$-probability simplex. $\mm$ is {\em neutral} if for any $\theta\in\Theta$ and any permutation $\sigma$ over $\ma$, there exists $\eta\in\Theta$ such that for all $R\in\ml(\ma)$, we have $\pi_\theta(R) = \pi_\eta(\sigma(R))$. 
\end{dfn}
See Example~\ref{ex:sapm} in Appendix~\ref{app:modelex} for two examples of single-agent preference models that correspond to the celebrated Mallows model and Plackett-Luce model, respectively.


\section{Smoothed Analysis Framework and The Main Technical Lemma}
\label{sec:maintech}
Many commonly-studied axioms and events in social choice, denoted by $X$, are defined based on per-profile properties in the following way. Let $r$ denote a voting rule or correspondence and let $P$ denote a profile. There is a function $\rel_X(r,P)\in \{0,1\}$ that indicates whether $X$ holds for $r$ at $P$. Then, $r$ satisfies $X$ if $\forall P, \rel_X(r,P)=1$, or equivalently, $\inf_P \rel_X(r,P)=1$. For example, for {\em anonymity}, let $\rel_\text{ano}(r,P)=1$ iff for all profiles $P'$ with $\hist(P')=\hist(P)$, $r(P')=r(P)$. For {\em neutrality}, let $\rel_\text{neu}(r,P)=1$ iff for all permutation $\sigma$ over $\ma$, $\sigma(r(P))=r(\sigma(P))$. For {\em non-existence of Condorcet cycle}, let $\ncc(P)=1$ iff there is no Condorcet cycle in $P$.

Our smoothed analysis framework assumes that each of the $n$ agents' preferences are chosen from a single-agent preference model $\mm=(\Theta,\ml(\ma),\Pi)$ by the adversary.
\begin{dfn}[\bf Smoothed likelihood of events]\label{dfn:smoothedlikelihood} Given a single-agent preference model $\mm=(\Theta,\ml(\ma),\Pi)$, $n\in\mathbb N$ agents, a function $\rel_X$ that characterizes an axiom or event $X$, and a voting rule  (or correspondence) $r$, the {\em smoothed likelihood} of $X$ is defined as $\inf_{\vec\pi\in \Pi^n}\expect_{P\sim {\vec\pi}} \rel_X(r,P)$.
\end{dfn}


For example, $\inf_{\vec\pi\in \Pi^n}\expect_{P\sim {\vec\pi}} \ncc(P)$ is the smoothed likelihood of non-existence of Condorcet cycle, which corresponds to the avoidance of Condorcet's paradox. $\inf_{\vec\pi\in \Pi^n}\Pr_{P\sim {\vec\pi}} (\rel_\text{ano}(P)+\rel_\text{neu}(P)=2)$ is the smoothed likelihood of satisfaction of anonymity$+$neutrality, which corresponds to the avoidance of the ANR impossibility.

We use the following simple example to show how to model events of interest as a system of linear constraints. The first event is closely related  to the smoothed Condorcet's paradox (Theorem~\ref{thm:smoothedCondorcet}) and the second event is closely related to the smoothed ANR theorem (Theorem~\ref{thm:smoothedANR}).

\begin{ex}\label{ex:cc} Let $m=3$ and $\ma = \{1,2,3\}$. For any profile $P$, let $x_{123}$ denote the number of $1\succ 2\succ 3$ in $P$. The event ``there is a Cordorcet cycle $1\ra 2\ra 3\ra 1$" can be represented by:
\begin{align}
(x_{213}+x_{231}+x_{321}) - (x_{123}+x_{132}+x_{312}) & < 0\label{eq:1->2}\\
(x_{312}+x_{321}+x_{132}) - (x_{231}+x_{213}+x_{123})&<0\label{eq:2->3}\\
(x_{123}+x_{132}+x_{213}) - (x_{312}+x_{321}+x_{231})&<0\label{eq:3->1}
\end{align}
Equation (\ref{eq:1->2}) (respectively, (\ref{eq:2->3}) and (\ref{eq:3->1})) states that $\umg(P)$ has edge $1\ra 2$ (respectively, $2\ra 3$ and $3\ra 1$). As another example, the event ``$\hist(P)$ is invariant to the permutation $\sigma$ over $\ma$ that exchanges $1$ and $2$'' can be represented by $\{x_{123} - x_{213} = 0, x_{132} - x_{231} = 0, x_{312} - x_{321} = 0\}$.
\end{ex}
Notice that in this example each constraint has the form $\vec E\cdot\vec x =0$ or $\vec S\cdot\vec x <0$, where $\vec E\cdot\vec 1 = 0$ and $\vec S\cdot\vec 1=0$. More generally, our main technical lemma upper-bounds the smoothed likelihood for the  Poisson multinomial variable $\hist(P)$ to satisfy a similar system of linear inequalities below.
\begin{dfn} Let $q, n\in \mathbb N$. For any vector of $n$ distributions $\vec\pi= (\pi_1,\ldots,\pi_n)$, each of which is over $[q]$, let $\vec Y=(Y_1,\ldots,Y_n)$ denote the vector of $n$ random variables distributed as $\pi_1,\ldots,\pi_n$, respectively, and  let $\vec X_{\vec \pi} =  \hist(\vec Y)$, i.e.~the Poisson multinomial variable that corresponds to $\vec Y$. 
\end{dfn}

\begin{dfn} \label{dfn:cab} 
Let 
$\cons(\vec x)=\{\ba\cdot (\vec x)^\top = (\vec 0)^\top \text{ and }\bb\cdot (\vec x)^\top < (\vec 0)^\top\}$, where $\ba$ is a $K\times q$ integer matrix that represents the {\em equations} and $\bb$ is an $L\times q$ integer matrix  with $K+L\ge 1$ that represents the {\em strict inequalities}.  Let $\rcons(\vec x)=\{\ba\cdot (\vec x)^\top = (\vec 0)^\top \text{ and }\bb\cdot (\vec x)^\top \le (\vec 0)^\top\}$ denote the relaxation of $\cons(\vec x)$.  Let $\ppoly{}$ and  $\ppolyz{}$  denote the solutions to $\cons(\vec x)$ and $\rcons(\vec x)$, respectively.  
\end{dfn}


\begin{lem}[\bf Main technical lemma]\label{lem:maintech} Let $q\in\mathbb N$ and let $\Pi$ be a closed set of strictly positive distributions over  $[q]$. Let $\conv(\Pi)$ denote the convex hull of $\Pi$. 

\text{\bf Upper bound.}  For any $n\in\mathbb N$ and any $\vec \pi\in \Pi^n$, 

$\hfill\Pr\left(\vec X_{\vec \pi} \in \sol\right)=\left\{\begin{array}{ll}0 &\text{if }\sol = \emptyset\\
\exp(-\Omega(n)) &\text{if }\sol \ne \emptyset \text{ and }\rsol\cap\conv(\Pi)= \emptyset\\
O(n^{-\frac{\rank(\ba)}{2}}) &\text{if }\sol \ne \emptyset \text{ and }\rsol\cap\conv(\Pi)\ne \emptyset
\end{array}\right.\hfill$
\text{\bf Tightness of the upper bound.}  There exists a constant $C$ such that for any $n'\in\mathbb N$, there exists $n'\le n\le Cn'$ and $\vec \pi\in\Pi^n$ such that

$\hfill
\Pr\left(\vec X_{\vec \pi} \in \sol\right)=\left\{\begin{array}{ll}
\exp(-O(n)) &\text{if }\sol \ne \emptyset \text{ and }\rsol\cap\conv(\Pi)= \emptyset\\
\Omega(n^{-\frac{\rank(\ba)}{2}}) &\text{if }\sol \ne \emptyset \text{ and }\rsol\cap\conv(\Pi)\ne \emptyset
\end{array}\right.\hfill$
\end{lem}
Lemma~\ref{lem:maintech} is quite general because the assumptions on $\Pi$ are mild and $\ba$ and $\bb$ are general enough to model a wide range of events in social choice as we will see later in the paper. It provides asymptotically tight upper bounds on the probability for the histogram of a randomly generated profile from $\vec\pi\in\Pi^n$ to satisfy all constraints in $\cons$. The bounds provide a trichotomy: if no vector satisfies all constraints in $\cons$, i.e.~$\sol = \emptyset$, then the upper bound is $0$; otherwise if $\sol \ne \emptyset$ and its relaxation $\rsol$ does not contain a vector in the convex hull of $\Pi$, then the upper bound is exponentially small; otherwise the upper bound is polynomially small in $n$, and the degree of polynomial is determined by the rank of $\ba$, specifically $-\frac{\rank(\ba)}{2}$. The tightness part of the lemma states that the upper bounds cannot be improved for all $n$.

At a high level the lemma is quite natural and follows the intuition of multivariate central limit theorem as follows. Roughly, $\vec X_{\vec \pi}$ is distributed  like a multinomial Gaussian whose expectation is $\vec \pi\cdot\vec 1 = \sum_{j=1}^n \pi_j$ (which is a $q$-dimensional vector). Then, the zero part of Lemma~\ref{lem:maintech} is trivial; the exponential part makes sense because the expectation $\vec \pi\cdot\vec 1$ is $\Theta(n)$ away from any vector in $\sol$; and the last part is expected to be $O(n^{-0.5})$ because the center $\vec \pi\cdot\vec 1$ satisfies the $\ba$ part of $\cons$. 

The surprising part of  the lemma is the degree of polynomial $-\frac{\rank(\ba)}{2}$ and its tightness. As discussed in the Introduction, all central limit theorems we are aware of are too coarse for proving the $O(n^{-\frac{\rank(\ba)}{2}})$ bound. To prove the polynomial upper bound, we introduce an alternative representation of $\vec X_{\vec \pi}$ to tackle the dependencies among its components, prove novel fine-grained concentration and anti-concentration bounds, focus on the reduced row echelon form of  $\ba$ plus an additional constraint on the total number of agents to characterize $\sol$, and then do a weighted counting of vectors that satisfy $\cons$. The full proof can be found in Appendix~\ref{sec:maintechproof}.


\section{Smoothed  Condorcet's Paradox and ANR (Im)possibility Theorem}
We first apply the main technical lemma (Lemma~\ref{lem:maintech}) to characterize the smoothed likelihood of Condorcet's paradox in the following dichotomy theorem, which holds for any fixed $m\ge 3$.
\begin{thm}[\bf Smoothed likelihood of Codorcet's paradox]\label{thm:smoothedCondorcet}
 Let $\mm= (\Theta,\ml(\ma),\Pi)$ be a strictly positive and closed single-agent preference model. 
 
\noindent{\bf Smoothed avoidance of Condorcet's paradox.} Suppose for all $\pi\in\conv(\Pi)$,  $\umg(\pi)$ does not contain a weak Condorcet cycle.  Then, for any $n\in\mathbb N$, we have:

$\hfill \inf_{\vec\pi\in \Pi^n}\expect_{P\sim \vec\pi} \ncc(P)= 1- \exp(-\Omega(n))\hfill$

\noindent{\bf Smoothed Condorcet's paradox.} Suppose there exists $\pi\in\conv(\Pi)$ such that  $\umg(\pi)$ contains a weak Condorcet cycle.  Then, there exist infinitely many $n\in\mathbb N$ such that:

$\hfill \inf_{\vec\pi\in \Pi^n}\expect_{P\sim \vec\pi} \ncc(P)= 1- \Omega(1)\hfill$
 
 \end{thm}
The smoothed avoidance  part of the theorem  is positive news: if there is no weak Condorcet cycle in the UMG of any distribution in the convex hull of $\Pi$, then no matter how the adversary sets agents' ground truth preferences, the probability for Condorcet's paradox to hold, which is $1-\inf_{\vec\pi\in \Pi^n}\expect_{P\sim \vec\pi} \ncc(P)=\sup_{\vec\pi\in \Pi^n}\Pr_{P\sim \vec\pi}(\ncc(P)=0)$, vanishes at an exponential rate  as $n\ra\infty$. Consequently, in such cases Arrows' impossibility theorem can be avoided because the pairwise majority rule satisfies all desired properties  mentioned in the theorem. The second part (smoothed paradox) states that otherwise the adversary can make Condorcet's paradox occur with constant probability. The proof is done by modeling  $\ncc(P)=0$ as systems of linear constraints as in Definition~\ref{dfn:cab}, each of which represents a target UMG with a weak Codorcet cycle as in Example~\ref{ex:cc},  then applying Lemma~\ref{lem:maintech}. The full proof is in Appendix~\ref{sec:smoothedCondorcetproof}.


We now turn to the smoothed ANR impossibility. 
We will reveal a relationship between all $n$-profiles and all permutation groups over $\ma$ after recalling some basic notions in group theory. The {\em symmetric group} over $\ma=[m]$, denoted by $\sgroup$, is the set of all permutations over $\ma$.  

\begin{dfn}
For any profile $P$, let $\perm(P)$ denote the set of all permutations $\sigma$ over $\ma$ that maps $\hist(P)$ to itself. Formally, $\perm(P)=\{\sigma \in \sgroup: \hist(P)=\sigma(\hist(P))\}$.
\end{dfn}

See Appendix~\ref{sec:prelimgroup} for additional notation and examples about group theory.\footnote{Some group theoretic notation and ideas in this paper are similar to those in a 2015 working paper by~\citet{Dogan2015:Anonymous} whose main results are different. See Appendix~\ref{sec:prelimgroup} for more details and discussions.} It is not hard to see that $\perm(P)$ is a permutation group.  
We now define a special type of permutation groups that ``cover'' all alternatives in $\ma$, which are closely related to the impossibility theorem.

\begin{dfn}For any permutation group $U\subseteq \sgroup$ and any alternative $a\in\ma$, we say that $U$ {\em covers} $a$ if there exists $\sigma\in U$ such that $a\ne \sigma(a)$. We say that $U$ {\em covers} $\ma$ if it covers all alternatives in $\ma$. For any $m$, let $\calU_m$ denote the set of all permutation groups that cover $\ma$.
\end{dfn} 

For example, when $m=3$, $\calU_3=\{\id,(1,2,3), (1,3,2)\}$, where $\id$ is the identity permutation and $(1,2,3)$ is the  circular permutation $1\ra2\ra3\ra1$. 
See Example~\ref{ex:permutationgroups} in Appendix~\ref{sec:prelimgroup} for the list of all permutation groups for $m=3$. In general $|\calU_m|>1$.
\begin{thm}[\bf Smoothed ANR (im)possibility]\label{thm:smoothedANR} Let $\mm= (\Theta,\ml(\ma),\Pi)$ be a strictly positive and closed single-agent preference model.  Let $\calU_m^\Pi= \{U\in \calU_m: \exists \pi\in \conv(\Pi), \forall \sigma\in U, \sigma(\pi)=\pi\}$, and when $\calU_m^\Pi\ne \emptyset$, let $l_{\min}= \min_{U\in  \calU_m^\Pi}|U|$ and $l_\Pi= \frac{l_{\min}-1}{l_{\min}}m!$.

\noindent{\bf Smoothed possibility.} There exist an anonymous voting rule $r_\text{ano}$ and a neutral voting rule $r_\text{neu}$ such that for any $r\in \{r_\text{ano}, r_\text{neu}\}$, any $n$, and any $\vec\pi\in \Pi^n$, we have:

$\hfill\Pr_{P\sim\vec\pi}(\rel_\text{ano}(r, P)+\rel_\text{neu}(r, P)<2)=\left\{\begin{array}{ll}O(n^{-\frac{l_\Pi}{2}})&\text{if } \calU_m^\Pi\ne\emptyset \\  \exp(-\Omega(n))&\text{otherwise}\end{array}\right.\hfill$

\noindent{\bf Smoothed impossibility.} For any voting rule $r$, there exist infinitely many $n\in\mathbb N$ such that:

$\hfill\sup_{\vec\pi\in\Pi^n}\Pr_{P\sim\vec\pi}(\rel_\text{ano}(r, P)+\rel_\text{neu}(r, P)<2)=\left\{\begin{array}{ll}\Omega(n^{-\frac{l_\Pi}{2}})&\text{if } \calU_m^\Pi\ne\emptyset \\  \exp(-O(n))&\text{otherwise}\end{array}\right.\hfill$
\end{thm}
Again, Theorem~\ref{thm:smoothedANR} holds for fixed $m\ge 3$. We note that for any profile $P$, $\rel_\text{ano}(r, P)+\rel_\text{neu}(r, P)<2$ if and only if at least one of anonymity or neutrality is violated at $P$. In other words, if $\rel_\text{ano}(r, P)+\rel_\text{neu}(r, P)=2$ then both anonymity and neutrality are satisfied at $P$. Therefore, the first part of Theorem~\ref{thm:smoothedANR} is called ``smoothed possibility'' because it states that 
no matter how the adversary sets agents' ground truth preferences, the probability for  $r_\text{ano}$ (respectively, $r_\text{neu}$)  to satisfy both anonymity and neutrality converges to $1$. 
The second part (smoothed impossibility) shows that the rate of convergence in the first part is asymptotically tight for all $n$. 
This is a mild impossibility theorem because violations of anonymity or neutrality may still vanish (at a slower rate) as $n\ra\infty$. 

The proof proceeds in the following three steps. Step 1.~For any $m$ and $n$, we  define a set of profiles, denoted by $\calT_{m,n}$, that represent the source of impossibility. In fact, $\calT_{m,n}$ is the set of all $n$-profiles $P$ such that $\perm(P)$ covers $\ma$, i.e.~$\perm(P)\in\calU_m$. 
Step 2.~To prove the smoothed possibility part, we define $r_\text{ano}$ and $r_\text{neu}$ that satisfy both anonymity and neutrality for all profiles that are not in $\calT_{m,n}$. Then, we apply Lemma~\ref{lem:maintech} to upper-bound the probability of $\calT_{m,n}$.  
Step 3.~The smoothed impossibility part is proved by applying the tightness part of Lemma~\ref{lem:maintech} to the probability of $\calT_{m,n}$.  The full proof can be found in Appendix~\ref{sec:proofsmoothedANR}.

In general $l_{\min}$ in Theorem~\ref{thm:smoothedANR} can be hard to characterize. The following lemma provides a lower bound on $l_{\min}$ by characterizing $\min_{U\in \calU_m}|U|$, whose group-theoretic proof is in Appendix~\ref{appendix:lemmatauproof}.
\begin{lem}\label{lem:taubounds}
For any $m\ge 2$, let  $l^* = \min_{U\in \calU_m}|U|$. We have  $l^* = 2$ if $m$ is even;  $l^* = 3$ if $m$ is odd and $3\mid m$; $l^* = 5$ if $m$ is odd, $3\nmid m$, and $5\mid m$; and $l^*=6$ for other $m$. 
\end{lem}
A notable special case of Theorem~\ref{thm:smoothedANR} is $\pi_\text{uni}\in\conv(\Pi)$, where $\pi_\text{uni}$ is the uniform distribution over $\ml(\ma)$. We note that for any permutation $\sigma$, $\pi_\text{uni} = \sigma( \pi_\text{uni})$, which means that  $\calU_m^\Pi = \calU_m$. Therefore,  only the polynomial bound in Theorem~\ref{thm:smoothedANR}  remains, with $l_\Pi=\frac{l^*-1}{l^*}m!$. In particular, $\pi_\text{uni}\in\conv(\Pi)$ for all neutral single-agent preference models  under IC, which corresponds to $\Pi=\{\pi_\text{uni}\}$.
See Corollary~\ref{coro:smoothedANR} in Appendix~\ref{appendix:ANRcoro} for the formal statement.

\section{Optimal Tie-Breaking for Anonymity $+$ Neutrality}
\label{sec:opttiebreaking}
While $r_\text{ano}$ and $r_\text{neu}$ in Theorem~\ref{thm:smoothedANR} are asymptotically optimal w.r.t.~anonymity + neutrality, they may be hard to compute. 
The following proposition shows that the commonly-used \lex{} and \fa{} mechanisms are far from being optimal for positional scoring rules.

%
%

\begin{prop}\label{prop:nonoptlexfa} Let $r$  be a voting rule obtained from a positional scoring correspondence by applying  \lex{} or \fa{}. Let $\mm= (\Theta,\ml(\ma),\Pi)$ be a strictly positive and closed single-agent preference model with $\pi_\text{uni}\in \conv(\Pi)$. There exist infinitely many $n\in\mathbb N$ such that:

$\hfill\sup_{\vec\pi\in\Pi^n}\Pr_{P\sim\vec\pi}\left(\rel_\text{ano}(r, P)+\rel_\text{neu}(r, P)<2\right)=\Omega(n^{- 0.5})\hfill$
\end{prop}
The proof is done by modeling ties under positional scoring correspondences as systems of linear constraints, then applying the tightness of the polynomial bound in Lemma~\ref{lem:maintech}. The full proof can be found in Appendix~\ref{sec:nonoptlexfaproof}. 
We now introduce a new class of easy-to-compute tie-breaking mechanisms that achieve the optimal upper bound $O(n^{- \frac{m!}{4}})$ in Theorem~\ref{thm:smoothedANR} when $m$ is even.


\begin{dfn}[\bf Most popular singleton ranking]  Given a profile $P$, we define its {\em most popular singleton ranking (\mpsr)} as $\mpsr(P) = \arg\max_{R}(P[R]: \nexists W\ne R \text{ s.t. } P[W] = P[R])$.
\end{dfn}
Put differently, a ranking $R$ is called a {\em singleton} in a profile $P$, if there does not exist another linear order that occurs for the same number of times in $P$. $\mpsr(P)$ is the singleton that occurs most frequently in $P$. If no singleton exists, then we let $\mpsr(P)=\emptyset$.  We now define tie-breaking mechanisms based on \mpsr{}.
\begin{dfn}[\bf \mpsr{} tie-breaking mechanism]\label{dfn:mpsrtb} For any voting correspondence $\cor$, any profile $P$, and any backup tie-breaking mechanism $\tb$, the {\em \mpsr-then-\tb} mechanism uses  $\mpsr(P)$ to break ties whenever $\mpsr(P)\ne \emptyset$; otherwise it uses $\tb$ to break ties.
\end{dfn}

\begin{thm}\label{thm:mpsropt} Let $\mm= (\Theta,\ml(\ma),\Pi)$ be a strictly positive and closed single-agent preference model with $\pi_\text{uni}\in \conv(\Pi)$. For any voting correspondence $\cor$ that satisfies anonymity and neutrality, let $r_\mpsr$ denote the voting rule obtained from $\cor$ by \mpsr{}-then-$\tb$. For any $n$ and any $\vec\pi\in \Pi^n$, 

$\hfill\Pr_{P\sim\vec\pi}(\rel_\text{ano}(r_\mpsr, P)+\rel_\text{neu}(r_\mpsr, P)<2)=O(n^{- \frac{m!}{4}})\hfill$

Moreover, if $\tb$ satisfies anonymity (respectively, neutrality) then so does $r_\mpsr$.
\end{thm}
The proof is done by showing that (1) anonymity and neutrality are preserved when $\mpsr(P)\ne \emptyset$, and (2) any profile $P$ with $\mpsr(P)= \emptyset$ can be represented by a system of linear constraints, whose smoothed likelihood is upper-bounded by the polynomial upper bound in Lemma~\ref{lem:maintech}. The full proof can be found in Appendix~\ref{sec:mpsroptproof}. 

Note that when $m$ is even, the $O(n^{- \frac{m!}{4}})$ upper bound in Theorem~\ref{thm:mpsropt} matches the optimal upper bound in light of Theorem~\ref{thm:smoothedANR} and Lemma~\ref{lem:taubounds}. This is good news because it implies that any anonymous and neutral correspondence can be made an asymptotically optimal voting rule w.r.t.~anonymity$+$neutrality by MPSR tie-breaking. When $m$ is odd, the $O(n^{- \frac{m!}{4}})$ upper bound in Theorem~\ref{thm:mpsropt} is suboptimal but still significantly better than that of the lexicographic or fixed-agent tie-breaking mechanism, which is $\Omega(n^{-0.5})$ (Proposition~\ref{prop:nonoptlexfa}). 

\section{Future work} 

We have only touched the tip of the iceberg of smoothed analysis in social choice. There are at least three major dimensions for future work: (1) other social choice axioms and impossibility theorems, for example Arrow's impossibility theorem~\cite{Kalai02:Fourier,Keller2010:On-the-probability,Mossel2012:A-quantitative} and the Gibbard-Satterthwaite theorem~\cite{Friedgut2011:A-quantitative,Mossel2015:A-quantitative}, (2) computational aspects in social choice~\cite{Brandt2016:Handbook,Baumeister2020:Towards} such as the smoothed complexity of winner determination and  complexity of manipulation, and (3) other social choice problems such as judgement aggregation~\cite{Nehama2013:Approximately,Filmus2019:AND-testing}, distortion~\cite{Procaccia06:Distortion,Anshelevich2018:Approximating,Mandal2019:Efficient,Kempe2020:Communication,Mandal2020:Optimal}, matching, resource allocation, etc.

\newpage

\section{Acknowledgements}
We thank Elliot Anshelevich, Rupert Freeman, Herve Moulin, Marcus Pivato, Nisarg Shah, Rohit Vaish, Bill Zwicker, participants of the COMSOC video seminar, and anonymous reviewers for helpful discussions and comments. This work is supported by NSF \#1453542, ONR \#N00014-17-1-2621, and a gift fund from Google.

\section*{Broader Impact} In this paper we aim to provide smoothed possibilities of social choice, which is an important problem in the society. Therefore, success of the research will benefit general public beyond the CS research community because better solutions are now available for a wide range of group decision-making scenarios.

{
\bibliographystyle{plainnat}
\bibliography{/Users/administrator/GGSDDU/references}
}

\newpage
\appendix
\section{Appendix: An Example of Single-Agent Preference Models}
\label{app:modelex}
\begin{ex}
\label{ex:sapm}
In a  {\em single-agent Mallows' model}  $\mm_\mallows$, $\Theta = \ml(\ma)\times [0,1]$, where in each $(R,\varphi)\in \Theta$, $R$ is  the central ranking and $\varphi$ is  the dispersion parameter. For any $W\in \ml(\ma)$, we have $\pi_{(R,\varphi)} = \varphi^{\kt(R,W)}/Z_{\varphi}$, where $\kt(R,W)$ is the {\em Kendall Tau distance} between $R$ and $W$, namely the number of pairwise disagreements between $R$ and $W$, and $Z_\varphi = \sum_{W\in\ml(\ma)}\varphi^{\kt(R,W)}$ is the normalization constant. For any $0<\underline{\varphi}\le 1$, we let $\mm_\mallows^{[\underline{\varphi},1]}$ denote the Mallows' model where the parameter space is  $\ml(\ma)\times [\underline{\varphi}, 1]$. 

As another example, in the {\em single-agent Plackett-Luce model} $\mm_\pl$, $\Theta= \{\vec \theta\in [0,1]^m: \vec \theta\cdot \vec 1 = 1\}$. For any $\vec \theta\in\Theta$ and any $R=\sigma(1)\succ \sigma(2)\succ\cdots\succ\sigma(m)$, we have $\pi_{\vec\theta} (R) = \prod_{i=1}^{m-1}\frac{\theta_{\sigma(i)}}{\sum_{l=i}^m \theta_{\sigma(l)}}$. For any $0<\underline{\varphi}\le 1$, we let $\mm_\pl^{[\underline{\varphi},1]}$ denote the Plackett-Luce model where $\Theta= \{\vec \theta\in [\underline{\varphi},1]^m: \vec \theta\cdot \vec 1 = 1\}$. 

It follows that for any $0<\underline{\varphi}\le 1$, $\mm_\mallows^{[\underline{\varphi},1]}$ and $\mm_\pl^{[\underline{\varphi},1]}$ are strictly positive, closed, and neutral.
\end{ex}

\section{Appendix: Proof of Lemma~\ref{lem:maintech}}
\label{sec:maintechproof}

\appLem{lem:maintech} {{\bf(Main technical lemma).} Let $q\in\mathbb N$ and $\Pi$ be a closed set of strictly positive distributions over  $[q]$. Let $\conv(\Pi)$ denote the convex hull of $\Pi$. 

\text{\bf Upper bound.}  For any $n\in\mathbb N$ and any $\vec \pi\in \Pi^n$, 
$$\Pr\left(\vec X_{\vec \pi} \in \sol\right)=\left\{\begin{array}{ll}0 &\text{if }\sol = \emptyset\\
\exp(-\Omega(n)) &\text{if }\sol \ne \emptyset \text{ and }\rsol\cap\conv(\Pi)= \emptyset\\
O(n^{-\frac{\rank(\ba)}{2}}) &\text{if }\sol \ne \emptyset \text{ and }\rsol\cap\conv(\Pi)\ne \emptyset
\end{array}\right.$$
\text{\bf Tightness of the upper bound.}  There exists a constant $C$ such that for any $n'\in\mathbb N$, there exists $n'\le n\le Cn'$ and $\vec \pi\in\Pi^n$ such that
$$
\Pr\left(\vec X_{\vec \pi} \in \sol\right)=\left\{\begin{array}{ll}
\exp(-O(n)) &\text{if }\sol \ne \emptyset \text{ and }\rsol\cap\conv(\Pi)= \emptyset\\
\Omega(n^{-\frac{\rank(\ba)}{2}}) &\text{if }\sol \ne \emptyset \text{ and }\rsol\cap\conv(\Pi)\ne \emptyset
\end{array}\right.$$
}

\begin{proof} The $\sol=\emptyset$ case trivially holds. Let $o=\rank(\ba)$. 
Let $\vec X_{\vec \pi} = \hist(\vec Y)=(X_{\vec \pi,1}\ldots,X_{\vec \pi,q})$. That is, for any $i\le q$, $X_{\vec \pi,i}$  represents the number of occurrences of outcome $i$ in $\vec Y$. For any $n\in\mathbb N$ and any $\vec \pi\in \Pi^n$, let $\vec \mu_{\vec\pi} =(\mu_{\vec\pi,1},\ldots,\mu_{\vec\pi,q}) = \expect(\sum _{j=1}^n \vec X_{\vec \pi}/n)$ denote the mean of $\vec X_{\vec \pi}/n$ and let $\vec\sigma_{\vec\pi}=(\sigma_{\vec\pi,1},\ldots,\sigma_{\vec\pi,q})$, where for each $i\le q$, $\sigma_{\vec\pi,i}=\sqrt {\text{Var}(X_{\vec \pi,i})/n}$. Because $\Pi$ is strictly positive, there exists $\epsilon_1 >0, \epsilon_2>0$ such that for all $n$, all $\vec \pi\in \Pi^n$, and all $i\le q$, we have $\epsilon_1< \mu_{\vec\pi,i}<\epsilon_2$ and $\frac{\epsilon_1}{\sqrt n}< \sigma_{\vec\pi,i}<\frac{\epsilon_2}{\sqrt n}$. 


\vspace{1mm}
\noindent{\bf \boldmath Upper bound when $\sol \ne \emptyset \text{ and }\rsol\cap\conv(\Pi)= \emptyset$.}  
It is not hard to see that $\rsol$ is convex and closed. 
Because $\Pi$ is closed and bounded, $\conv(\Pi)$ is convex, closed and compact. Because $\rsol\cap \conv(\Pi)=\emptyset$, by the strict hyperplane separation theorem, there exists a hyperplane that strictly separates  $\rsol$ and $\conv(\Pi)$. Therefore, there exists $\epsilon'>0$ such that for any $\vec x_1\in \rsol$ and any $\vec x_2\in \conv(\Pi)$, we have $|\vec x_1-\vec x_2|_\infty>\epsilon'$, where $|\cdot|_\infty$ is the $L_\infty$ norm.  This means that any  solution to $\conv(\vec x)$ is at least $\epsilon' n$ away from $n\cdot \vec \mu_{\vec \pi}$ in $L_\infty$. Therefore, we have:
\begin{align*}
&\Pr\left(\vec X_{\vec \pi} \in \sol\right) \le \Pr\left( |\vec X_{\vec \pi} - n\cdot \vec \mu_{\vec \pi}|_\infty> \epsilon'n\right) \le \sum_{i=1}^q \Pr(|X_{\vec \pi,i} -n\mu_{\pi,i}|>\epsilon' n)\\
&\le 2q\exp\left(-\frac{(\epsilon')^2 n }{(1-2\epsilon)^2}\right)
\end{align*}
The last inequality follows after  Hoeffding's inequality (Theorem 2 in~\cite{Hoeffding63:Probability}), where $\epsilon$ is a constant such that any distribution in $\Pi$ is above $\epsilon$.

\vspace{1mm}
\noindent{\bf \boldmath Upper bound when $\sol \ne \emptyset \text{ and }\rsol\cap\conv(\Pi)\ne \emptyset$.} Let $\consa(\vec x)=\{\ba\cdot\vec x = (\vec 0)^\top\}$ denote the relaxation of $\cons(\vec x)$ by removing the $\bb$ part. Let $\consa(\vec X_{\vec \pi})$ denote the event that $\vec X_{\vec \pi}$ satisfies all constraints in $\consa$. It follows that each vector in $\sol$ is a solution to $\consa(\vec x)$, which means that $\Pr(\vec X_{\vec \pi}\in \sol)\le \Pr(\consa(\vec X_{\vec \pi}))$. Therefore, it suffices to prove that for any $n\in\mathbb N$ and any $\vec \pi\in \Pi^n$,  $\Pr(\consa(\vec X_{\vec \pi})) = O(n^{-\frac{o}{2}})$. 
Because $\ba\cdot (\vec 1)^\top =(\vec 0)^\top$, it follows that $\vec 1 = \{1\}^q$ is linearly independent with the row vectors of $\ba$. Therefore, the rank of $\ba'=\left[\begin{array}{c}\ba\\\vec 1\end{array}\right]$ is $o+1$. Let $\consap(\vec x)=\{\ba\cdot \vec x = 0 \text{ and } \vec 1\cdot \vec x=n\}$.
From basic linear algebra, in particular the reduced row echelon form (a.k.a.~row canonical form)~\cite{Meyer2000:Matrix} of $\ba'$ computed by Gauss-Jordan elimination where $n$ is treated as a constant, we know that there exist $I_0\subseteq [q]$ and $I_1\subseteq [q]$ such that $I_0\cap I_1=\emptyset$, $|I_0|=\rank(\ba')=\rank(\ba)+1$, and an $|I_0|\times ( |I_1|+1)$ matrix $\bd$ in $\mathbb Q$ such that $\consap(\vec x)$ is equivalent to $(\vec x_{I_0})^\top = \bd\cdot [\vec x_{I_1}, n]^\top$, where $\vec x_{I_0}$ is the subvector of $\vec x$ that contains variables whose subscripts are in $I_0$. In other words, a vector $\vec x$ satisfies $\consap(\vec x)$ if and only if $(\vec x_{I_0})^\top = \bd\cdot [\vec x_{I_1}, n]^\top$. See Example~\ref{ex:bordacondorcetD} in Appendix~\ref{appendix:maintech} for an example of deriving $\bd$ from $\ba'$.

We note that $I_1\cup I_0 = [q]$. For the sake of contradiction suppose this is not true, which means that the reduced row echelon form of $\ba'$ has a column of zeros for some variable $x_j$. However, this means that $\vec 1$ is not a linear combination of the rows of the reduced row echelon form of $\ba'$, which is a contradiction because $\ba'$, which includes $\vec 1$, can be obtained from a series of linear transformations on its reduced row echelon form. W.l.o.g.~in the remainder of this proof we let $I_0 = \{1,\ldots, o+1\}$ and $I_1= \{o+2,\ldots, q\}$. 

The hardness in bounding $\Pr(\consa(\vec X_{\vec \pi}))$ is that elements of $\vec X_{\vec \pi}$ are not independent, and typical asymptotical tools such as  Lyapunov-type bound are too coarse. To solve this issue, we use the following alternative representation of $Y_1,\ldots,Y_n$. 
For each $j\le n$, we use a binary random variable $Z_j\in\{0,1\}$ to represent whether the outcome of $Y_j$ is in $I_0$ (corresponding to $Z_j=0$) or is in $I_1$ (corresponding to $Z_j=1$). Then, we use another random variable $W_j\in [q]$ to represent the outcome of $Y_j$ conditioned on $Z_j$. See Figure~\ref{fig:alt1} for an illustration.
\begin{figure}[htp]
\centering
\includegraphics[width=.7\textwidth]{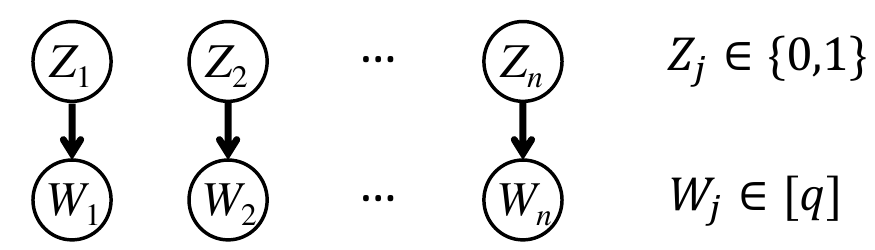}
\caption{The representation of $\vec Y$ as $\vec Z$ and $\vec W$, where $\vec W$ has the same distribution as $\vec Y$.\label{fig:alt1}}
\end{figure}

 At a high level, this addresses the independent issue because components of $\vec W$ are conditionally independent given $\vec Z$, and as we will see below, concentration happens in $\vec W$ when $\vec Z$ contains $\Theta(n)$ many $0$'s.

\begin{dfn}[\bf \boldmath  Alternative representation of $Y_1,\ldots,Y_{n}$]\label{dfn:altfory} For each $j\le n$, we define a Bayesian network with two random variables $Z_j \in \{0,1\}$ and $W_j\in [q]$, where $Z_j$ is the parent of $W_j$, and 
\begin{itemize}
\item for each $l\in \{0,1\}$, $\Pr(Z_j = l) = \Pr(Y_j \in I_l)$; 
\item   for each $l\in \{0,1\}$ and each $t\le q$, $\Pr(W_j = t|Z_j=l) = \Pr(Y_j = t|Y_j\in I_l)$.
\end{itemize}
In particular, if $t\not\in I_l$ then $\Pr(W_j = t|Z_j=l)=0$.
\end{dfn}
It follows that $W_j$ has the same distribution as $Y_j$. For any $\vec z\in \{0,1\}^n$, we let $\ind_0(\vec z)\subseteq [n]$ denote the indices of $\vec z$ that equals to $0$. Given $\vec z$, we let $\vec W_{\ind_0(\vec z)}$ denote  the set of all $W_j$'s with $z_j = 0$, and let $\hist(\vec W_{\ind_0(\vec z)})$ denote the vector of $o+1$ random variable that correspond to the histogram of $\vec W_{\ind_0(\vec z)}$. Similarly,  we let $\hist(\vec W_{\ind_1(\vec z)})$ denote  the vector of $q-o-1$ random variables that correspond to the histogram of $\vec W_{\ind_1(\vec z)}$. Recall that each $\pi_j$ is at least $\epsilon>0$. We have the following calculation on $\Pr(\consa(\vec X_{\vec \pi}))$, i.e.~the probability that $\hist(\vec Y)$ satisfies all constraints in $\consa$.
\begin{align}
&\Pr(\consa(\vec X_{\vec \pi})) =\sum_{\vec z\in \{0,1\}^n}\Pr(\vec Z = \vec z) \Pr\left(\consa(\hist(\vec W))\;\middle\vert\; \vec Z = \vec z \right) (\text{total probability})\notag\\
=&\sum_{\vec z\in \{0,1\}^n}\Pr(\vec Z = \vec z) \Pr\left(\hist(\vec W_{\ind_0(\vec z)})^\top = \bd\cdot [\hist(\vec W_{\ind_1(\vec z)}), n]^\top \;\middle\vert\; \vec Z = \vec z \right)\notag\\
=&\sum_{\vec z\in \{0,1\}^n}\Pr(\vec Z = \vec z)\sum_{\vec x\in {\mathbb Z_{\ge 0}^{q-o-1}}}\Pr\left(\hist(\vec W_{\ind_1(\vec z)}) =\vec x \;\middle\vert\; \vec Z = \vec z \right)\notag
\\
&\hspace{40mm}\times \Pr\left(\hist(\vec W_{\ind_0(\vec z)})^\top =\bd\ \cdot [\vec x, n]^\top \;\middle\vert\; \vec Z = \vec z \right)\label{eq:condz}
\end{align}
\begin{align}
=&\sum_{\vec z\in \{0,1\}^n}\Pr(\vec Z = \vec z)\sum_{\vec x\in {\mathbb Z_{\ge 0}^{q-o-1}}}\Pr\left(\hist(\vec W_{\ind_1(\vec z)}) =\vec x \;\middle\vert\; [\vec Z]_{\ind_1(\vec z)} = \vec 1\right)\notag
\\
&\hspace{40mm}\times \Pr\left(\hist(\vec W_{\ind_0(\vec z)})^\top =\bd\ \cdot [\vec x, n]^\top \;\middle\vert\; [\vec Z]_{\ind_0(\vec z)} = \vec 0\right)\label{eq:condz1}
\end{align}
\begin{align}
\le &\sum_{\vec z\in \{0,1\}^n: [\hist(\vec z)]_0\ge 0.9 \epsilon n} \Pr(\vec Z = \vec z) \sum_{\vec x\in {\mathbb Z_{\ge 0}^{q-o-1}}}\Pr\left(\hist(\vec W_{\ind_1(\vec z)}) =\vec x \;\middle\vert\; [\vec Z]_{\ind_1(\vec z)} = \vec 1\right)\notag\\
&\times \Pr\left(\hist(\vec W_{\ind_0(\vec z)})^\top =\bd\cdot [\vec x, n]^\top \;\middle\vert\; [\vec Z]_{\ind_0(\vec z)} = \vec 0\right) + \Pr([\hist(\vec Z)]_0 < 0.9\epsilon n)\label{eq:histz}
\end{align}
where $[\hist(\vec z)]_0$ is the number of $0$'s in $\vec z$. (\ref{eq:condz}) holds because $W_j$'s are independent of each other given $Z_j$'s, which means that for any $\vec z\in\{0,1\}^n$, $\hist(\vec W_{\ind_0(\vec z)})$ and $\hist(\vec W_{\ind_1(\vec z)})$ are independent given $\vec z$. (\ref{eq:condz1}) holds because 
$\vec W_{\ind_1(\vec z)}$ (respectively, $\vec W_{\ind_0(\vec z)}$) is independent of $[\vec Z]_{\ind_0(\vec z)}$ (respectively, $[\vec Z]_{\ind_1(\vec z)}$) given $[\vec Z]_{\ind_1(\vec z)}$ (respectively, $[\vec Z]_{\ind_0(\vec z)}$)  in light of independence in the Bayesian network.

To simplify notation, we write (\ref{eq:histz}) as follows.
\begin{align*}
&\sum_{\vec z\in \{0,1\}^n: [\hist(\vec z)]_0\ge 0.9 \epsilon n} \Pr(\vec Z = \vec z) \sum_{\vec x\in {\mathbb Z_{\ge 0}^{q-o-1}}} F_1(\vec z,\vec x)\times F_2(\vec z,\vec x) +F_3, \text{ where}\\
&F_1(\vec z,\vec x)=\Pr\left(\hist(\vec W_{\ind_1(\vec z)}) =\vec x \;\middle\vert\; [\vec Z]_{\ind_1(\vec z)} = \vec 1\right)\\
&F_2(\vec z,\vec x)=\Pr\left(\hist(\vec W_{\ind_0(\vec z)})^\top =\bd\cdot [\vec x, n]^\top \;\middle\vert\; [\vec Z]_{\ind_0(\vec z)} = \vec 0\right)\\
&F_3= \Pr([\hist(\vec Z)]_0 < 0.9\epsilon n)
\end{align*}

We now show that  given $[\hist(\vec z)]_0\ge 0.9 \epsilon n$, for any $\vec x\in \mathbb Z_{\ge 0}^{q-o-1}$, 
$$F_2(\vec z,\vec x)=O((0.9\epsilon n)^{-\frac{o}{2}})= O(n^{-\frac{o}{2}}),$$ 
which follows after the following lemma, where $n=|\ind_0(\vec z)|$ and  $q^*=o+1$. The lemma can be seen as a Poisson multivariate extension of the Littlewood-Offord-Erd\H{o}s anti-concentration bound, because it says that the probability for $\hist(\vec W_{\ind_0(\vec z)})^\top$ to take a specific value $\bd\cdot [\vec x, n]^\top$ (which is a constant vector given $\vec x$) is bounded above by $O(n^{-\frac{o}{2}})$.

\begin{lem}[\bf Point-wise anti-concentration bound for Poisson multinomial variables]\label{lem:upperboundp} Given $q^*\in\mathbb N$ and $\epsilon>0$. There exists a constant $C^*>0$ such that for any $n\in \mathbb N$ and any vector $\vec Y' = (Y'_1,\ldots,Y'_n)$ of $n$ independent random variables over $[q^*]$, each of which is above $\epsilon$,  and any vector $\vec x \in {\mathbb Z^{q^*}_{\ge 0}}$, we have $\Pr(\hist(\vec Y')=\vec x)<C^*n^{\frac{1-q^*}{2}}$.
\end{lem} 

\begin{proof} When $\vec x\cdot \vec 1\ne n$ the inequality holds for any $C^*>0$. Suppose $\vec x\cdot \vec 1= n$, we prove the claim by induction on $q^*$. When $q^*=1$ the claim holds for any $C_1^*>1$. Suppose the claim holds for $q^*=q'-1$ w.r.t.~constant $C_{q'-1}^*$. W.l.o.g.~suppose $x_1\ge x_2\cdots\ge x_{q^*}$. When $q^*=q'$, we use the following representation of $Y'_1,\ldots,Y'_n$ that is similar to Definition~\ref{dfn:altfory}. For each $j\le n$, let $Y'_j$ be represented by the Bayesian network with two random variables $Z'_j\in \{0,1\}$ and its child $W'_j\in \{1,\ldots, q'\}$. Let  $\Pr(Z'_j=0)=\Pr(Y'_j\in \{1,\ldots,q'-1\})$ and $\Pr(Z'_j=1)=\Pr(Y'_j=q')$ . Let $\Pr(W'_j=q'|Z'_j=1) = 1$ and for all $l\le q'-1$,  $\Pr(W'_j=l|Z'_j=0) = \Pr(Y'_j = l|Y'_j\in \{1,\ldots q'-1\})$.

It follows that $\Pr(Z'_j=0,W'_j=q')=0$, $\Pr(Y'_j=q') = \Pr(Z'_j=1,W'_j=q')$, and for all $l\le q'-1$, $\Pr(Y'_j=l) = \Pr(Z'_j=0,W'_j=l)$ and $\Pr(Z'_j=1,W'_j=l)=0$. In other words, $Z'_j$ determines whether $Y'_j\in \{1,\ldots, q'-1\}$ (corresponding to $Z'_j=0$) or $Y'_j=q'$ (corresponding to $Z'_j=1$), and $W'_j$ determines the value of $Y'_j$ conditioned on $Z'_j$. 
By the law of total probability, we have:
 \begin{align*}
&\Pr(\hist(\vec Y')=\vec x)  =\sum_{\vec z\in \{0,1\}^n: \vec z\cdot \vec 1 = x_{q'}}\Pr\left(\hist(\vec W') = \vec x\;\middle\vert\; \vec Z' = \vec z\right)\cdot \Pr(\vec Z' = \vec z),
\end{align*}
We note that $W_j'$s are independent of each other given $\vec Z'$, and $\Pr(W_j=q'| Z_j=1)=1$. Therefore, for any $\vec z\in \{0,1\}^n$ with $\vec z\cdot \vec 1 = x_{q'}$, we have 
\begin{align}
&\Pr\left(\hist(\vec W') = \vec x\;\middle\vert\; \vec Z' = \vec z\right)=\Pr\left( \hist(\vec W')_{-q'} = \vec x_{-q'}\;\middle\vert\; \vec Z'_{\ind_0(\vec z)} = \vec 0\right)\notag\\
 \le&  C_{q'-1}^*|n-x_{q'}|^{(2-q')/2}\le  C_{q'-1}^* (\frac{q'-1}{q'})^{(2-q')/2} n^{(2-q')/2}\notag 
\end{align}
where $\hist(\vec Z')_{-q'}$ is the subvector of $\hist(\vec Z')$ by taking out the $q'$-th component and $\vec x_{-q'} = (x_1,\ldots,x_{q'-1})$.  
The first inequality follows after the induction hypothesis, because each $W'_j$ with $j\in \ind_0(\vec z)$ is a random variable over $\{1,\ldots, q'-1\}$ that is above $\epsilon$. The second step uses the assumption that $x_1\ge x_2\ge\cdots\ge x_{q'}$, which means that $x_{q'}\le \frac{1}{q'}n$. The next claim, which can be seen as an extension of the Littlewood-Offord-Erd\H{o}s  anti-concentration bound to Poisson binomial distributions, proves that $\Pr(\vec Z'\cdot \vec 1=x_{q'})=O(n^{-1/2})$.

\begin{claim} There exists a constant $C'$ that does not depend on $n$ or $q'$ such that $\Pr(\vec Z'\cdot \vec 1=x_{q'}) \le C'n^{-1/2}$.
\end{claim} 
\begin{proof}For any $j\le n$, recall that $\Pr(Z'_j = 0)\ge \epsilon$ and $\Pr(Z'_j = 1)\ge \epsilon$. Therefore, for all $j\le n$, $\text{Var}(Z'_j)\ge \epsilon(1-\epsilon)$ and there exists a constant $\rho$ such that for all $j\le n$, $\expect(|Z'_j-\expect(Z'_j)|^3)<\rho$. Let $\mu=\sum_{j=1}^{n}Z'_j$ and $\sigma = {\sqrt{\sum_{j=1}^{n}\text{Var}(Z'_j)}}$. We have $\sigma\ge \sqrt{n\epsilon(1-\epsilon)}$. By  Berry-Esseen theorem (see e.g.~\cite{Durrett91:Probability}), there exists a constant $C_0$ such that:
\begin{align*}
&\Pr(\vec Z'\cdot \vec 1=x_{q'}) \le \Pr(x_{q'}-1< \vec Z'\cdot \vec 1\le x_{q'}+1)= 
\Pr(\frac{x_{q'}-1}{\sigma}\le \frac{\vec Z'\cdot \vec 1}{\sigma}\le \frac{x_{q'}+1}{\sigma})\\
\le & (\hist(\frac{x_{q'}-\mu+1}{\sigma})-\hist(\frac{x_{q'}-\mu-1}{\sigma})) + 2C_0(n\epsilon(1-\epsilon))^{-3/2}n\rho\\
\le & \frac{2}{\sigma} +2C_0(n\epsilon(1-\epsilon))^{-3/2}n\rho \le  \left(\frac{2}{\sqrt{\epsilon(1-\epsilon)}}+\frac{2C_0\rho}{(\epsilon(1-\epsilon))^{3/2}}\right)n^{-1/2}
\end{align*}
The claim follows by letting $C'=\frac{2}{\sqrt{\epsilon(1-\epsilon)}}+\frac{2C_0\rho}{(\epsilon(1-\epsilon))^{3/2}}$.
\end{proof}
Finally, we have 
\begin{align*}
&\Pr(\hist(\vec Y')=\vec x) \le C_{q'-1}^*(\frac{q'-1}{q'})^{(2-q')/2}n^{(2-q')/2} \sum_{\vec z\in \{0,1\}^n: \vec z\cdot \vec 1 = x_{q'}} \Pr(\vec Z' = \vec z)\\
 \le& C'n^{-1/2}C_{q'-1}^*(\frac{q'-1}{q'})^{(2-q')/2}n^{(2-q')/2} 
\end{align*}
This proves the $q^*=q'$ case by letting $C_{q'}^*= C'C_{q'-1}^*(\frac{q'-1}{q'})^{(2-q')/2}$. This completes the proof of Lemma~\ref{lem:upperboundp}.\end{proof}

Because random variables in $\vec Y$ are above $\epsilon$, for all $j\le n$, $Z_j$ takes $0$ with probability at least $\epsilon$. Therefore, $\expect(\vec Z\cdot \vec 1) \ge \epsilon n$. By Hoeffding's inequality, $F_3 = \Pr([\hist(\vec Z)]_0 < 0.9\epsilon n)$ is exponentially small in $n$, which means that it is $O(n^{-\frac{o}{2}})$. We also note that for any $\vec z$ and $\vec x$ we have $F_1(\vec z,\vec x)\le 1$. Therefore, continuing (\ref{eq:histz}), we have:
\begin{align*}
&\Pr(\consa(\vec X_{\vec \pi})) 
\le \sum_{\vec z\in \{0,1\}^n: [\hist(\vec z)]_0\ge 0.9 d n} \Pr(\vec Z = \vec z) \sum_{\vec x\in {\mathbb Z_{\ge 0}^{q-o-1}}} F_1(\vec z,\vec x)\times F_2(\vec z,\vec x) +F_3\\
\le &\sum_{\vec z\in \{0,1\}^n: [\hist(\vec z)]_0\ge 0.9 d n} \Pr(\vec Z = \vec z) O(n^{-\frac{o}{2}}) + O(n^{-\frac{o}{2}}) = O(n^{-\frac{o}{2}})
\end{align*}
This proves the upper bound when $\sol \ne \emptyset \text{ and }\rsol\cap\conv(\Pi)\ne \emptyset$.

\vspace{1mm}
\noindent{\bf \boldmath  Tightness of the upper bound when $\sol \ne \emptyset \text{ and }\rsol\cap\conv(\Pi)= \emptyset$.}  We first prove the following claim that will be frequently used in this proof. The claim follows after a straightforward application of Theorem~1(i) in~\cite{Cook86:Sensitivity}. It states that for any $\vec x\in \sol$ and any $l>0$, $\sol$ contains an integer vector that is close to $l\cdot \vec x$.
\begin{claim}\label{claim:existenceintegersolution} Suppose $\sol\ne\emptyset$. There exists a constant $C'>0$ such that for any $\vec x\in\sol$, there exists an integer vector $\vec x^*\in\sol$ with $|\vec x^* - l\cdot \vec x|_\infty<C'$.
\end{claim}
\begin{proof} We first prove that $\sol$ contains an integer vector. Because  $\sol\ne\emptyset$ and $\ba$ is an integer matrix, $\sol$ contains a rational solution $\vec x'$, which can be chosen from the neighborhood of any vector in $\sol$. It follows that for any $L>0$, $L\cdot\vec x'\in \sol$ because $\ba \cdot (L\cdot\vec x')^\top = (\vec 0)^\top$ and $\bb \cdot (L\cdot\vec x')^\top < (\vec 0)^\top$. It is not hard to see that there exists $L\in\mathbb N$ such that $L\cdot\vec x'\in\mathbb Z_{\ge 0}^q$. 

The claim then follows after Theorem~1(i) in~\cite{Cook86:Sensitivity} by replacing $\bb \cdot (\vec x)^\top < (\vec 0)^\top$ with equivalent constraints $\bb \cdot (L\cdot\vec x')^\top \le (-\vec 1)^\top$.
\end{proof}
Let $C'$ denote the constant in Claim~\ref{claim:existenceintegersolution}. W.l.o.g.~let $\vec x\in\sol$ denote an arbitrary solution such that $\vec x>C'\cdot \vec 1$; otherwise we consider $\vec x+(|\vec x|_\infty +C')\cdot \vec 1\in \sol$.  For any $l\in \mathbb N$, let $\vec x_l^*$ denote the integer vector in $\sol$ that is guaranteed by Claim~\ref{claim:existenceintegersolution};  let $n_l = \vec x_l^*\cdot \vec 1$; and let $\vec \pi^l\in \Pi^{n_l}$ denote an arbitrary vector of $n_l$ distributions, each of which is chosen from $\Pi$.  Let $\vec y^l\in [q]^{n_l}$ denote an arbitrary vector with $\hist(\vec y^l)=\vec x_l^*$.  We note that for any $l\in \mathbb N$, each distribution in $\vec \pi^l$ is above $\epsilon$. Therefore, $\Pr(\vec X_{\vec \pi^l} = \vec y^l) \ge \epsilon^{n_l} = \exp(n_l\log \epsilon)$, which is $\exp(- O(n_l))$.  It is not hard to verify that for any $l\in\mathbb N$, $\vec x_l^*$ is strictly positive and  $\frac{n_{l+1}}{n_l}$ is bounded above by a constant, denoted by $\hat C$. Let $C=\max(n_1,\hat C)$. This proves the tightness of the upper bound when $\sol \ne \emptyset \text{ and }\rsol\cap\conv(\Pi)= \emptyset$.

\vspace{1mm}
\noindent{\bf \boldmath  Tightness of the upper bound when $\sol \ne \emptyset \text{ and }\rsol\cap\conv(\Pi) \ne \emptyset$.} Let $\vec x^*\in \rsol\cap \conv(\Pi)$ and write
$\vec x^* = \sum_{i=1}^k \alpha_i \pi_i$  as a linear combination of vectors in $\Pi$.  The tightness of the upper bound will be proved in the following four steps. Step~1.~For every $l\in \mathbb N$ we prove that there exists an integer vector $\vec y^l\in\sol$ that is $\Theta(\sqrt l)$ away from $l\cdot \vec x^*$. Then, we define $\vec \pi^l = (\vec \pi_1^l,\ldots,\vec \pi_k^l)\in \Pi^{n_l}$, where for all $i\le k$, $\vec \pi_i^l$ is approximately $l\alpha_i$ copies of $\pi_i$, which means that $n_l = \Theta(l)$.  Step~2.~We  identify $\Omega(n^{\frac{(q-1)(k-1)+q-o-1}{2}})$ combinations of values of $\hist(\vec X_{\vec \pi_1^l}),\ldots,\hist(\vec X_{\vec \pi_k^l})$, such that the sum of each such combination is no more than $O(\sqrt n_l)$ away from $\vec y^l$ and  is a solution to $\cons(\vec x)$. Step~3.~We  prove that the probability of each such combination is $\Omega(n_l^{\frac{(1-q)}{2}})$ by applying Lemma~\ref{lem:probclose} below. Finally, we will have $\Pr(\vec X_{\vec \pi}\in \sol)\ge  \Omega(n_l^{\frac{(q-1)(k-1)+q-o-1}{2}})\times  \Omega(n_l^{\frac{(1-q)}{2}})^k = \Omega(n_l^{-\frac{o}{2}})$.

{\bf Step~1.} Let $\vec x^{\#} \in \sol$ and $\vec x^*\in \rsol\cap \conv(\Pi)$. W.l.o.g.~suppose $\vec x^{\#}$ is strictly positive;  otherwise let $l>0$ denote an arbitrary number such that $\vec x^{\#}+l\cdot\vec 1$ is strictly positive, and because $\ba\cdot (\vec 1)^\top=(\vec 0)^\top$ and $\bb\cdot (\vec 1)^\top=(\vec 0)^\top$, we have $\vec x^{\#}+l\cdot\vec 1\in\sol$. Because $\vec x^*\in \rsol\cap \conv(\Pi)$, we can write  $\vec x^* = \sum_{i=1}^k \alpha_i \pi_i$, where for all $i\le k$, $\alpha_i>0$ and $\pi_i\in \Pi$, and $\sum_{i=1}^k \alpha_i = 1$.  We note that $\vec x^*\ge \epsilon\cdot \vec 1$, because $\Pi$ is above $\epsilon$. 

For any sufficiently large $l\in\mathbb N$, we let $\vec y^{l\#}\in \sol$ denote the integer approximation to $\sqrt l \cdot \vec x^{\#}$ that is guaranteed by Claim~\ref{claim:existenceintegersolution}; let  $\vec y^{l*}$ denote the integer approximation to $l\cdot \vec x^*$ that is guaranteed by  Claim~\ref{claim:existenceintegersolution}, where we merge all  rows $\vec B$ in $\bb$ with $\vec B\cdot \vec x^*=0$ to $\ba$ before applying the claim. It follows that $\vec y^{l*}\in \rsol$. Let $C$ denote an arbitrary constant such that $|\vec y^{l\#}-\sqrt l \cdot \vec x^{\#}|_\infty<C$ and $|\vec y^{l*}-l \cdot \vec x^{*}|_\infty<C$.

Let $\vec y^l = \vec y^{l\#} +\vec y^{l*}$ and let $n_l=\vec y^l\cdot \vec 1$. It follows that $\vec y^l$ is a strictly positive integer solution to $\cons(\vec x)$ and each of its element is $\Theta(n_l)$. We now define $\vec \pi^l$ that  is approximately $l$ copies of $\vec x^{*}$. Formally, for each $i\le k-1$, let $\vec \pi_i^l$ denote the vector of $\beta_i=\lfloor l \alpha_i\rfloor$ copies of $\pi_i$. Let $\vec \pi_k^l$ denote the vector of $\beta_k=n_l-\sum_{i=1}^{k-1} \beta_i$ copies of $\pi_k$. It follows that for any $i\le k-1$, $|\beta_i-l \alpha_i| \le 1$, and $|\beta_k-l \alpha_k|\le k + \vec y^{l\#}\cdot \vec 1  = O(\sqrt l) = O(\sqrt n_l)$. Let $\vec \pi^l = (\vec \pi_1^l,\ldots,\vec \pi_k^l)$ denote the vector of $n_l$ distributions. It follows that  $|\expect(\vec X_{\vec \pi^l})-\vec y^l|_\infty = O(\sqrt n_l)$.

{\bf Step~2.} We define a set $\mx_l\subseteq {\mathbb N}^{qk}$ of vectors $(\vec x_1,\ldots,\vec x_k)$ such that $\vec x_i= (x_{i1},\ldots,x_{iq})\in {\mathbb N}^q$ will be used as a target value for $\vec X_{\vec \pi_i^l}$ soon in the proof. 
Let $\mx_l$ denote the set of all integer vectors $(\vec x_1,\ldots,\vec x_k)$ that satisfy the following three conditions.
\begin{itemize}
\item [
(i)] Let $\zeta>0$ be a constant whose value will be specified later. For any $i\le k-1$ and any $j\le q-1$, we require $|x_{ij}-\pi_{ij}^l \beta_i|< \zeta\sqrt n_l$, where $\vec \pi_i^l= (\pi_{i1}^l,\ldots,\pi_{iq}^l)$.  Therefore, $|x_{iq}-\pi_{iq}^l\beta_i|< (q-1)\zeta\sqrt n_l$, because $\sum_{j\le q}x_{ij} = \beta_i$.

\item[(ii)] Let $\rho$ be the least common multiple of the denominators of all entries in $\bd$. For example, in Example~\ref{ex:bordacondorcetD} we have $\rho=2$.   For each $o+2\le j\le q$, we require $|x_{kj}-(y_j^l-\sum_{i=1}^{k-1}x_{ij})|< \zeta\sqrt n_l$ and $\rho$ divides $x_{kj}-(y_j^l-\sum_{i=1}^{k-1}x_{ij})$.

\item[(iii)] $\vec x =\sum_{i=1}^k\vec x_{i}$ is a solution to $(\vec x_{I_0})^\top = \bd\cdot [\vec x_{I_1}, n_l]^\top$.
\end{itemize}

Each element $(\vec x_1,\ldots,\vec x_k)$ of $\mx_l$ can be generated in the following three steps. First, $\vec x_1,\ldots, \vec x_{k-1}$ are arbitrarily chosen according to condition (i) above. Second, $x_{k(o+2)},\ldots,x_{kq}$ are chosen according to condition (ii) above given $\vec x_1,\ldots, \vec x_{k-1}$. This guarantees that for each $o+2\le j\le q$, $\sum_{i=1}^k x_{ij}$ is no more than $O(\sqrt n_l)$ away from $y_j^l$ and $\sum_{i=1}^k x_{ij}-y_{j}^l$ is divisible by $\rho$. Finally, $x_{k1},\ldots,x_{k(o+1)}$ are determined  by condition (iii) together with other components of $\vec x$ that are specified in the first and second step. More precisely, 
$$\left[\begin{array}{c}x_{k1}\\ \vdots \\ x_{k(o+1)}\end{array}\right] = \bd\cdot \left[\begin{array}{c}\sum_{i=1}^kx_{i(o+2)}\\ \vdots \\ \sum_{i=1}^k x_{iq}\\n_l\end{array}\right] -\sum_{i=1}^{k-1}\left[\begin{array}{c}x_{i1}\\ \vdots \\ x_{i(o+1)}\end{array}\right]$$

$x_{k1},\ldots,x_{k(o+1)}$ are integers because for all $o+2\le j\le q$, $\rho$ divides $y_j-\sum_{i=1}^kx_{ij}$ and $\bd\cdot [\vec y_{I_1},n_l]^T = [\vec y_{I_0},n_l]^T$. 
We let $\zeta>0$ be a sufficiently small constant that does not depend on $n_l$, such that the following two conditions hold. 

\begin{itemize}
\item [(1)] Each $(\vec x_1,\ldots,\vec x_k)\in\mx_l$ is strictly positive. This can be achieved by assigning $\zeta$ a small positive value, because $|\sum_{i=1}^k \vec x_i-\vec y^l|_\infty=O(\zeta)\sqrt n_l$ and each element in $\vec y^l$ is strictly positive and is $\Theta(n_l)$. 

\item [(2)] For any $(\vec x_1,\ldots,\vec x_k)\in\mx_l$, we have $\sum_{i=1}^k\vec x_i\in \sol$. Let $\vec x = \sum_{i=1}^k\vec x_i$. By definition we have $\ba\cdot (\vec x)^\top=(\vec 0)^\top$. We also have 
\begin{align*}
&\bb\cdot (\vec x)^\top = \bb\cdot (\vec x-\vec y^l)^\top+\bb\cdot (\vec y^l)^\top= \bb\cdot (\vec x-\vec y^l)^\top+\bb\cdot (\vec y^{l\#} +\vec y^{l*})^\top\\
\le &\bb\cdot (\vec x-\vec y^l)^\top+\bb\cdot (\vec y^{l\#})^\top
\end{align*}
The inequality follows after recalling that $\vec y^{l*}\in \rsol$, which means that $\bb\cdot (y^{l*})^\top\le (\vec 0)^\top$.
Notice that $|\sum_{i=1}^k \vec x_i-\vec y^l|_\infty=O(\zeta)\sqrt n_l$ and each element in $\bb\cdot (\vec y^{l\#})^\top$ is $\Theta(n_l)$. Therefore, when $\zeta$ is sufficiently small we have  $\bb\cdot (\vec x)^\top<(\vec 0)^\top$. This means that when $\zeta>0$ is sufficiently small we have $\vec x\in \sol$.

\end{itemize}
%

For any $l$ with  $\frac{\zeta\sqrt n_l}{\rho}>1$, we have: 

$$|\mx_l|\ge (\frac{1}{\rho})^{q-o-1}\zeta^{(q-1)(k-1)+q-o-1} n_l^{((q-1)(k-1)+q-o-1)/2}=\Omega(n_l^{((q-1)(k-1)+q-o-1)/2}).$$

This is because according to condition (i) there are at least $(\zeta\sqrt n_l)^{(q-1)(k-1)}$ combinations of values for $\vec x_1,\ldots,\vec x_{k-1}$, and according to condition (ii) there are at least $(\frac{\zeta\sqrt n_l}{\rho})^{q-o-1}$ combinations of values for $x_{k(o+2)},\ldots, x_{kq}$. 

{\bf Step~3.}  By the definition of $\mx_l$, there exists a constant $\alpha>0$ that does not depend on $n_l$ such that for each  $(\vec x_1,\ldots,\vec x_k)\in\mx_l$  and each $i\le k$, we have $|\vec x_i-\beta_i\pi_i|_\infty<\alpha\sqrt n_l$.  Also because all agents' preferences are independently generated, we have: $\Pr(\forall i\le k, \vec X_{\vec \pi_i} = \vec x_i) = \prod_{i=1}^k \Pr(\vec X_{\vec \pi_i}= \vec x_i)$.

We note that for each $i\le k$, $\vec X_{\vec \pi_i}$ is the histogram of $\beta_i$ i.i.d.~random variables, each of which is distributed as $\pi_i$. The next lemma implies that for each $i\le k$, $\Pr(\vec X_{\vec \pi_i}= \vec x_i)$ is $\Omega(n_l^{(1-q)/2})$. 

\begin{lem}[\bf Point-wise concentration bound for i.i.d.~Poisson multinomial variables]\label{lem:probclose}  Given $q\in\mathbb N$,  $\epsilon>0$, $\alpha>0$. There exists a constant $\beta>0$ such that for any distribution $\pi$ over $[q]$ that is above $\epsilon$, any $n\in\mathbb N$, and any vector $\vec x\in \mathbb Z_{\ge 0}^q$ with $\vec x\cdot \vec 1 = n$ and $|\vec x - n\pi|_\infty<\alpha\sqrt n$, we have $\Pr(\vec X_\pi=\vec x)>\beta n^{\frac{1-q}{2}}$, where $\vec X_{\pi}$ is the Poisson multinomial variables corresponding to $n$ i.i.d.~random variables, each of which is distributed as $\pi$.
\end{lem}
\begin{proof} Let $\vec x = (x_1,\ldots,x_q)$, and let $\vec d = \vec x - n\pi$. We have: 
\begin{align} 
& \Pr(\vec X_\pi=\vec x) = {n \choose x_1} {n - x_1\choose x_2} \cdots {x_{q-1}+x_q\choose x_{q-1}}\prod_{i=1}^q\pi_i^{x_i} = \frac{n!}{\prod_{i=1}^q x_i!}\prod_{i=1}^q\pi_i^{x_i}\notag\\
\ge & \frac{\frac{1}{\sqrt {2 {\lambda} n}}(n/e)^n}{\prod_{i=1}^q (\frac{1}{\sqrt{2 \lambda x_i}} e^{1/12}(x_i/e)^{x_i})}\prod_{i=1}^q\pi_i^{x_i} = Cn^{\frac{1-q}{2}}\frac{n^n}{\prod_{i=1}^q x_i^{x_i}}\prod_{i=1}^q \pi_i^{x_i}\label{step:stirling}\\
=&Cn^{\frac{1-q}{2}} \prod_{i=1}^q\left(\frac{n\pi_i}{x_i}\right)^{n\pi_i}\prod_{i=1}^q \left(\frac{n\pi_i}{x_i}\right)^{x_i-n\pi_i} \ge Cn^{\frac{1-q}{2}}\prod_{i=1}^q \left(\frac{n\pi_i}{x_i}\right)^{d_i}\label{step:entropy}
\end{align}
Inequality (\ref{step:stirling}) is due to Robbins' Stirling approximation~\cite{Robbins1955:Remark}, where $\lambda$ denotes ratio of circumference to diameter ($\pi$ has already been used to denote a distribution). $C$ is a constant that does not depend on $n$. Inequality (\ref{step:entropy}) is because $\prod_{i=1}^q\left(\frac{n\pi_i}{x_i}\right)^{n\pi_i} = \exp\left(nD_\text{KL}(\pi\|\frac{\vec x}{n})\right)$, where $D_\text{KL}(\pi\|\frac{\vec x}{n})$ is the KL divergence of $\frac{\vec x}{n}$ from $\pi$, which is non-negative, meaning that  $\prod_{i=1}^q\left(\frac{n\pi_i}{x_i}\right)^{n\pi_i}\ge 1$. Let $\pi_{\min} = \min_{i} \pi_i\ge \epsilon.$
\begin{align*}
&\prod_{i=1}^q \left(\frac{n\pi_i}{x_i}\right)^{d_i} = \prod_{i=1}^q \frac{1}{(1+\frac{d_i}{n\pi_i})^{d_i}} \ge \prod_{i\le q: d_i<0} \frac{1}{(1+\frac{d_i}{n\pi_i})^{d_i}}  \ge  
\left((1-\frac{\alpha\sqrt n}{n\pi_{\min}} )^{\alpha\sqrt n}\right)^{q}\\
=&\left((1-\frac{\alpha}{\pi_{\min}\sqrt n} )^{\frac{\pi_{\min}\sqrt n}{\alpha}}\right)^{q\alpha^2/\pi_{\min}}\end{align*}
As $\lim_{x\ra\infty}(1-\frac{1}{x})^x = \frac{1}{e}$, when $n$ is large enough we have $\prod_{i=1}^q \left(\frac{n\pi_i}{x_i}\right)^{d_i} > (\frac{1}{2e})^{q\alpha^2/\pi_{\min}}$ which is a constant that does not depend on $n$. This proves that $ \Pr(\vec X_\pi=\vec x)  = \Omega(n^{\frac{1-q}{2}})$.
\end{proof}

By Lemma~\ref{lem:probclose}, we have $\Pr(\vec X_{\vec \pi_i}= \vec x_i)=\Omega(\beta_i^{(1-q)/2})= \Omega(n_l^{(1-q)/2})$.

{\bf Finally,} we have $\Pr(\vec X_{\vec \pi}\in \sol)\ge  \Omega(n_l^{\frac{(q-1)(k-1)+q-o-1}{2}})\times  \Omega(n_l^{\frac{(1-q)}{2}})^k = \Omega(n_l^{-\frac{o}{2}})$. Note that we require $l$ to be sufficiently large such that  $\frac{\zeta\sqrt n_l}{\rho}>1$ in order to guarantee that $|\mx_l|$ is large enough.  Because $n_l = \Theta(l)$, there exists a constant $\hat C$ such that for any $l\in\mathbb N$, $\frac{n_{l+1}}{n_l}<\hat C$. Let  $C=\max(n_{L+1},\hat C)$. This proves the tightness of the upper bound when $\sol \ne \emptyset \text{ and }\rsol\cap\conv(\Pi) \ne \emptyset$.
\end{proof}

\subsection{ Examples}
\label{appendix:maintech}

\begin{ex}\label{ex:bordacondorcetm=3} Let $m=3$ and $\ma = \{1,2,3\}$. For any profile $P$, let $x_{123}$ denote the number of  $1\succ 2\succ 3$ in $P$. The event ``alternatives $1$ and $2$ are co-winners under Borda as well as the only two weak Condorcet winner'' can be represented by the following constraints.
\begin{align}
2(x_{123}+x_{132})+x_{213}+x_{312}&=2(x_{213}+x_{231})+x_{123}+x_{321}\label{eq:borda12}\\
x_{123}+x_{132}+x_{312}&=x_{213}+x_{231}+x_{321}\label{eq:pairwise12}\\
2(x_{312}+x_{321})+x_{132}+x_{231}&<2(x_{123}+x_{132})+x_{213}+x_{312}\label{eq:borda13}\\
x_{312}+x_{321}+x_{231}&< x_{123}+x_{132}+x_{213}&\label{eq:pairwise13}\\
x_{312}+x_{321}+x_{132}&<x_{231}+x_{213}+x_{123}&\label{eq:pairwise23}
\end{align}
Equation~(\ref{eq:borda12}) states that the Borda scores of $1$ and $2$ are the same; equation (\ref{eq:pairwise12}) states that $1$ and $2$ are tied in their head-to-head competition; inequality (\ref{eq:borda13}) states that the Borda sore of $1$ is strictly higher than the Borda score of $3$; and inequalities (\ref{eq:pairwise13}) and (\ref{eq:pairwise23}) require that $1$ and $2$ beat $3$ in their head-to-head competitions, respectively. 
\end{ex}

\begin{ex}\label{ex:bordacondorcetD} We show how to obtain $\bd$ using the setting in Example~\ref{ex:bordacondorcetm=3}. Let $\vec x = [x_{123}, x_{132},x_{213},x_{231},x_{312},x_{321}]$. We have $\ba = \left[\begin{array}{r r r r r r}
1& 2 &-1&-2&1&-1\\1& 1 &-1&-1&1&-1\end{array}\right]$ and $\ba'= \left[\begin{array}{c}
\ba\\\vec 1\end{array}\right]$. The reduced echelon form of $\ba'$ and its corresponding equations can be calculated as follows.
\begin{align*}
&\left[\begin{array}{r r r r r r|c}
1& 2 &-1&-2&1&-1&0\\1& 1 &-1&-1&1&-1&0\\
1&1&1&1&1&1& n\end{array}\right]\xrightarrow{R1; R2-R1;R3-R1}\left[\begin{array}{r r r r r r|c}
1& 2 &-1&-2&1&-1&0\\0& -1 &0&1&0&0&0\\
0&-1&2&3&0&2& n\end{array}\right]\\
&\xrightarrow{R1+2R2;-R2;R3-R2} \left[\begin{array}{r r r r r r|c}
1& 0 &-1&0&1&-1&0\\0& 1 &0&-1&0&0&0\\
0&0&2&2&0&2& n\end{array}\right]\\
&\xrightarrow{R1+R2/2;R2;R3/2} \left[\begin{array}{r r r r r r|c}
1& 0 &0&1&1&0&\frac{n}{2}\\0& 1 &0&-1&0&0&0\\
0&0&1&1&0&1& \frac{n}{2}\end{array}\right]
\end{align*} The text above each arrow represents matrix operations. $R1,R2,R3$ represents the first, second, and the third row vector of the matrix on the left. For example, ``$R1; R2-R1;R3-R1$'' represents that the in the right matrix, the first row is R1 of the left matrix; the second row is obtained by subtracting $R1$ from $R2$; and the third row is obtained from subtracting $R1$ from $R3$. Let $I_0= [x_{123}, x_{132},x_{213}]$ and $I_1=[x_{231},x_{312},x_{321}]$. We have:
$$\bd =  \left[\begin{array}{r r r r  }
-1& -1 &0&\frac{1}{2}\\1& 0&0&0\\
-1&0&-1 & \frac{1}{2}\end{array}\right]\text{ and }\left[\begin{array}{c}x_{123}\\ x_{132}\\x_{213}\end{array}\right] = \bd \times \left[\begin{array}{c}x_{231}\\x_{312}\\x_{321}\\n\end{array}\right] $$\end{ex}

\section{Appendix: Proof of Theorem~\ref{thm:smoothedCondorcet}}
\label{sec:smoothedCondorcetproof}

\appThm{thm:smoothedCondorcet}{{\bf (Smoothed likelihood of Codorcet's paradox).}
 Let $\mm= (\Theta,\ml(\ma),\Pi)$ be a strictly positive and closed single-agent preference model. 
 
\vspace{1mm}
\noindent{\bf Smoothed avoidance of Condorcet's paradox.} Suppose for all $\pi\in\conv(\Pi)$,  $\umg(\pi)$ does not contain a weak Condorcet cycle.  Then, for any $n\in\mathbb N$, we have:

$\hfill \inf_{\vec\pi\in \Pi^n}\expect_{P\sim \vec\pi} \ncc(P)= 1- \exp(-\Omega(n))\hfill$

\vspace{1mm}
\noindent{\bf Smoothed Condorcet's paradox.} Suppose there exists $\pi\in\conv(\Pi)$ such that  $\umg(\pi)$ contains a weak Condorcet cycle.  Then, there exist infinitely many $n\in\mathbb N$ such that:

$\hfill \inf_{\vec\pi\in \Pi^n}\expect_{P\sim \vec\pi} \ncc(P)= 1- \Omega(1)\hfill$
}
\begin{proof}
The  theorem is proved by applying Lemma~\ref{lem:maintech} multiple times, where $\cons$ represents the profiles whose UMG contains a specific weak Codorcet cycle. Formally, we have the following definitions.
\begin{dfn}[\bf Variables and pairwise constraints]\label{dfn:varcons} For any linear order $R\in\ml(\ma)$, let $x_R$ be a variable that represents the number of times $R$ occurs in a profile. Let $\mx_\ma=\{x_R:R\in\ml(\ma)\}$ and let $\vec x_\ma$ denote the vector of elements of $\mx_\ma$ w.r.t.~a fixed order. For any pair of different alternatives $a,b$, let $\pair_{a,b}(\vec x_\ma)$ denote the linear combination of variables in $\mx_\ma$, where for any $R\in\ml(\ma)$, the coefficient of $x_R$ is $1$ if $a\succ_R b$; otherwise the coefficient is $-1$. 
\end{dfn}
For any profile $P$, $\pair_{a,b}(\hist(P))$ is the weight on the $a\succ b$ edge in $\wmg(P)$. It is not hard to check that $\pair_{a,b}(\vec 1)=0$. 

\begin{dfn}\label{dfn:cg}
For any unweighted directed graph $G$ over $\ma$, we define $\pcons{G}$ to be the constraints as in Definition~\ref{dfn:cab} that is based on matrices $\pba{G}$ that represents $\{\pair_{a,b}(\vec x_\ma)=0: (a,b)\not\in G \text{ and } (b,a)\not\in G\}$ and $\pbb{G}$ that represents $\{\pair_{b,a}(\vec x_\ma)<0: (a,b)\in G\}$. Let $\ppoly{G}$ and $\ppolyz{G}$ denote the solutions to $\pcons{G}$ and its relaxation $\overline{C}_G$ as in Definition~\ref{dfn:cab}, respectively. 
\end{dfn}
We immediate have the following claim.
\begin{claim}\label{claim:solumg}
For any profile $P$ and any unweighted directed graph $G$ over $\ma$, $\hist(P)\in \ppoly{G}$ if and only if $G=\umg(P)$; $\hist(P)\in {\ppolyz{G}}$ if and only if $\umg(P)$ is a subgraph of $G$. Moreover, $\rank(\pba{G})$ equals to the number of ties in $G$.
\end{claim}
\begin{proof}
The necessary and sufficient conditions for $\hist(P)\in \ppoly{G}$ and $\hist(P)\in {\ppolyz{G}}$ follow after their definitions. Because the number of equations in $\pba{G}$ equals to the number of ties in $G$, it suffices to prove that the equations in $\pba{G}$ are independent. This is true because for any pair of alternatives $(a,b)$, the UMG of the following profile of two rankings only contains one edge $[a\ra b$: $a\succ b\succ\text{others}]$ and $[\text{rev}\succ a\succ b]$,  where $\text{rev}$ is the reverse order of other alternatives.
\end{proof}
\noindent{\bf Proof of the smoothed avoidance part.} Let $\mc$ denote the set of all $\pcons{G}$, where $G$ is an unweighted directed graph without weak Condorcet cycles.  We have the following observations. 
\begin{itemize}
\item[(1)] For any profile $P$, $\ncc(P)=0$ if and only if there exists $G\in\mc$ such that $\hist(P)\in \ppoly{G}$. To see this, when $\ncc(P)=0$, we have $\hist(P)\in \ppoly{\umg(P)}$ and $\umg(P)\in\mc$; and vice versa, if $\hist(P)\in \ppoly{G}$ for some $G\in\mc$, then $\ncc(P)=0$. 
\item [(2)] For any graph $G$, $\ppoly{G}\ne\emptyset$, which follows after McGarvey's theorem~\cite{McGarvey53:Theorem}. 
\item [(3)] The total number of UMGs over $\ma$ only depends on $m$ not on $n$, which means that $|\mc|$ can be seen as a constant that does not depend on $n$. 
\end{itemize}

The three observations imply that for any distribution of $P$, we have 
$$\Pr(\ncc(P)=0)\le \sum_{G\in \mc}\Pr(\hist(P)\in \ppoly{G})$$
Therefore, based on observation (3) above, to prove the smoothed avoidance part of the theorem, it suffices to prove that for any $n$, any $\vec\pi\in\Pi^n$, and any $G\in\mc$, we have
$$\Pr_{P\sim \vec\pi}(\hist(P)\in \ppoly{G})=\exp(-\Omega(n))$$ 

This follows after the exponential case in Lemma~\ref{lem:maintech} applied to $\pcons{G}$. To see this, we first note that $\ppoly{G}\ne\emptyset$ according to observation (2) above. Also for each $\pi\in\Pi$, because $\umg(\pi)$ does not contain a weak Condorcet cycle, $\umg(\pi)$ is not a subgraph of $G$, which contains a Condorcet cycle.  Therefore, $\hist(\pi)\not \in {\ppolyz{G}}$ due to Claim~\ref{claim:solumg}, which means that  $\ppolyz{G}\cap \conv(\pi)= \emptyset$.

\noindent{\bf Proof of the  smoothed paradox part.} Let $\pi\in\conv(\Pi)$ denote any distribution such that $\umg(\pi)$ contains a weak Condorcet cycle. Let $G$ denote an arbitrary supergraph of $\umg(\pi)$ that contains a Condorcet cycle, e.g.~by completing the weak Condorcet cycle in  $\umg(\pi)$. The  weak smoothed paradox part is proved by applying the tightness of the polynomial case  in Lemma~\ref{lem:maintech}  to $\ppoly{G}$ following a similar argument with the proof for the smoothed avoidance part. 
%
\end{proof}

\section{Additional Preliminaries and Examples of Group Theory}
\label{sec:prelimgroup}
After this paper is accepted by NeurIPS, we discovered that a 2015 working paper by~\citet{Dogan2015:Anonymous} has already used similar notation and ideals to provide an alternative proof for Moulin's elegant characterization of ANR impossibility~\cite[Problem~1]{Moulin1983:The-Strategy} as well as obtaining a new characterization on ANR impossibility for social welfare functions (that outputs a ranking over $\ma$). Their main results are quite different from the smoothed ANR impossibility in this paper. 

In Appendix~\ref{sec:moulin}, we provide an alternative proof to \citet{Moulin1983:The-Strategy}'s characterization of ANR impossibility. At a high level the proof idea is similar to that by~\citet{Dogan2015:Anonymous} though the details appear different as far as we can tell. We do not claim the group theoretic approach nor the proof in Appendix~\ref{sec:moulin} contributions of this paper, but still include Appendix~\ref{sec:moulin} for information and convenience in case a reader is curious about the proof using notation in this paper.

The {\em symmetric group} over $\ma=[m]$, denoted by $\sgroup$, is the set of all permutations over $\ma$. A permutation $\sigma$ that maps each $a\in \ma$ to $\sigma(a)$ can be represented in two ways.
\begin{itemize}
\item {\em Two-line form:} $\sigma$ is represented by a $2\times m$ matrix, where the first row is $(1,2,\ldots, m)$ and the second row is $(\sigma(1),\sigma(2),\ldots, \sigma(m))$. 
\item {\em Cycle form:} $\sigma$ is represented by non-overlapping cycles over $\ma$, where each cycle $(a_1,\cdots, a_k)$ represent $a_{i+1}=\sigma(a_i)$ for all $i\le k-1$, and with $a_{1}=\sigma(a_k)$. 
\end{itemize}
It follows that any cycle in  the cycle form $(a_1,\cdots, a_k)$ is equivalent to $(a_2,\cdots, a_k, a_1)$. 
Following the convention, in the cycle form $a_1$ is the smallest elements in the cycle. For example, all permutations in $S_3$ are represented in two-line form and cycle form respective in the Table~\ref{tab:s3}.
\begin{table}[htp]
\centering
\caption{${\cal S}_{[3]}$ where $m=3$.\label{tab:s3}}
\begin{tabular}{|@{\small\ }c|@{} c@{}|@{}c@{}|@{}c@{}|@{}c@{}|@{}c@{}|@{}c@{}|}
\hline
Two-line&$\left(\begin{array}{ccc} 1& 2 & 3 \\ 1 & 2 & 3
\end{array}\right)$ & $\left(\begin{array}{ccc} 1& 2 & 3 \\ 2 &1 & 3
\end{array}\right)$& $\left(\begin{array}{ccc} 1& 2 & 3 \\ 1 & 3 & 2
\end{array}\right)$& $\left(\begin{array}{ccc} 1& 2 & 3 \\ 3 & 2 & 1
\end{array}\right)$& $\left(\begin{array}{ccc} 1& 2 & 3 \\ 2 & 3 & 1
\end{array}\right)$& $\left(\begin{array}{ccc} 1& 2 & 3 \\ 3 & 1 & 2
\end{array}\right)$ \\
\hline
Cycle& $()$ or $\id$& $(1,2)$& $(2,3)$& $(1,3)$& $(1,2,3)$& $(1,3,2)$ \\
\hline
\end{tabular}
\end{table}

A {\em permutation group} $G$ is a subgroup of $\sgroup$ where the identity element is the identity permutation $\id$, and for any $\sigma, \eta\in \sgroup$, $\sigma\circ \eta$ is the permutation where for any linear order $V\in \ml(\ma)$, $(\sigma\circ \eta) (V) = \sigma( \eta(V))$. For example, $(1,2)\circ (2,3) = (1,2,3)$, because $1\succ 2\succ 3$ is first mapped to $1\succ 3\succ 2$ by permutation $(2,3)$, then to $2\succ 3\succ 1$ by permutation $(1,2)$. There are six subgroups of ${\cal S}_{[3]}$: $\{\id\}$, $\{\id, (1,2)\}$, $\{\id, (2,3)\}$, $\{\id, (1,3)\}$, $\{\id, (1,2,3), (1,3,2)\}$, and  ${\cal S}_{[3]}$.

\begin{ex}\label{ex:permutationgroups}
 Table~\ref{tab:perm} shows all permutation groups over $\ms_{[3]}$ as the result of $\perm(P)$ for some profiles.
\begin{table}[htp]
\centering
\caption{\small Examples of $\perm(P)$, where $123$ represents $1\succ 2\succ 3$.  \label{tab:perm}}
\begin{tabular}{|@{\hspace{0.3mm}}c@{\hspace{0.3mm}}|@{\hspace{0.3mm}}c@{\hspace{0.3mm}}|@{\hspace{0.3mm}}c@{\hspace{0.3mm}}|@{\hspace{0.3mm}}c@{\hspace{0.3mm}}|@{\hspace{0.3mm}}c@{\hspace{0.3mm}}|@{\hspace{0.3mm}}c@{\hspace{0.3mm}}|@{\hspace{0.3mm}}c@{\hspace{0.3mm}}|@{\hspace{0.3mm}}c@{\hspace{0.3mm}}|}
\hline
$P$ & $P[123]$ & $P[132]$ & $P[213]$ & $P[2 3 1]$ & $P[312]$ & $P[321]$ & $\perm(P)$\\
\hline
$P_1$ & $1$ & $2$ & $2$ & $2$ & $2$ & $2$ & $\{\id\}$ \\
\hline
$P_2$ & $3$ & $5$ & $3$ & $5$ & $4$ & $4$ & $\{\id, (1,2)\}$ \\
\hline
$P_3$ & $3$ & $5$ & $4$ & $4$ & $5$ & $3$ & $\{\id, (1,3)\}$ \\
\hline
$P_4$ & $3$ & $3$ & $5$ & $4$ & $5$ & $4$ & $\{\id, (2,3)\}$ \\
\hline
$P_5$ & $3$ & $5$ & $5$ & $3$ & $3$ & $5$ & $\{\id, (1,2,3), (1,3,2)\}$ \\
\hline
$P_6$ & $1$ & $1$ & $1$ & $1$ & $1$ & $1$ & $\ms_{[3]}$ \\
\hline
\end{tabular}
\end{table}

$P_1$ in the table  consists of one ranking for $1\succ2\succ3$ and two rankings for each of the remaining five linear orders. $\perm(P_1)$ only contains $\id$ because if it contains any other permutation $\sigma$, then we must have $P[1\succ 2\succ 3]=P[\sigma(1\succ 2\succ 3)]$, which is impossible. $(12)\in \perm(P_2)$ because $P_2[1\succ 2\succ 3] = P_2[2\succ 1\succ 3]=3$, $P_2[1\succ 3\succ 2] = P_2[2\succ 3\succ 1]=5$, and $P_2[3\succ 1\succ 2] = P_2[3\succ 2\succ 1]=4$. It is not hard to verify that no other permutations except $\id$ belong to $\perm(P_2)$.
\end{ex}

\section{Proof of Theorem~\ref{thm:smoothedANR}}
\label{sec:proofsmoothedANR}

\appThm{thm:smoothedANR}{\bf (Smoothed ANR (im)possibility theorem).}{ Let $\mm= (\Theta,\ml(\ma),\Pi)$ be a strictly positive and closed single-agent preference model.  Let $\calU_m^\Pi= \{U\in \calU_m: \exists \pi\in \conv(\Pi), \forall \sigma\in U, \sigma(\pi)=\pi\}$, and when $\calU_m^\Pi\ne \emptyset$, let $l_{\min}= \min_{U\in  \calU_m^\Pi}|U|$ and $l_\Pi= \frac{l_{\min}-1}{l_{\min}}m!$.

\vspace{1mm}
\noindent{\bf Smoothed possibility.} There exist an anonymous voting rule $r_\text{ano}$ and a neutral voting rule $r_\text{neu}$ such that for any $r\in \{r_\text{ano}, r_\text{neu}\}$, any $n$, and any $\vec\pi\in \Pi^n$, we have:

$\hfill\Pr_{P\sim\vec\pi}(\rel_\text{ano}(r, P)+\rel_\text{neu}(r, P)<2)=\left\{\begin{array}{ll}O(n^{-\frac{l_\Pi}{2}})&\text{if } \calU_m^\Pi\ne\emptyset \\  \exp(-\Omega(n))&\text{otherwise}\end{array}\right.\hfill$

\vspace{1mm}
\noindent{\bf Smoothed impossibility.} For any voting rule $r$, there exist infinitely many $n\in\mathbb N$ such that:

$\hfill\sup_{\vec\pi\in\Pi^n}\Pr_{P\sim\vec\pi}(\rel_\text{ano}(r, P)+\rel_\text{neu}(r, P)<2)=\left\{\begin{array}{ll}\Omega(n^{-\frac{l_\Pi}{2}})&\text{if } \calU_m^\Pi\ne\emptyset \\  \exp(-O(n))&\text{otherwise}\end{array}\right.\hfill$
}

\begin{proof} The proof is done in three steps.

{\bf Step 1.~Identifying the source of the ANR impossibility}. 
\begin{dfn}\label{dfn:tmn}
For any $n\in\mathbb N$ and any $m\ge 2$, let $\calT_{m,n}$ denote the set of $n$-profiles $P$ such that $\perm(P)\in {\calU_m}$. That is,
$$\calT_{m,n} = \{P\in\ml(\ma)^n: \perm(P)\text{ covers }\ma\}$$
\end{dfn}
The following lemma states that profiles in $\calT_{m,n}$ are the intrinsic source of the ANR impossibility.

\begin{lem}\label{lem:sumdegree} For any voting rule $r$, any $n\in\mathbb N$, any $m\ge 2$, and any $P\in \calT_{m,n}$, we have:
\begin{equation}
\label{equ:anoneu} 
\rel_\text{ano}(r,P)+\rel_\text{neu}(r,P)\le 1
\end{equation}
\end{lem}
\begin{proof} We first partition $\calT_{m,n}$ to $\hist^{-1}({H}_1)\cup\cdots \cup  \hist^{-1}({H}_J)$, where $H_1,\ldots,H_J$ are some histograms of $n$ votes. The partition exists because for any $P\in \calT_{m,n}$ and any profile $P'$ with $\hist(P')=\hist(P)$, we have  $\perm(P')=\perm(P)\in {\calU_m}$, which means that $P'\in \calT_{m,n}$. For any $j\le J$, if there exist $P_1,P_2\in \hist^{-1}({H}_j)$ such that $r(P_1)\ne r(P_2)$, then for all $P\in \hist^{-1}({H}_j)$ we have $\rel_\text{ano}(P)=0$, which means that inequality (\ref{equ:anoneu}) holds for all $P\in \hist^{-1}({H}_j)$. Otherwise if $r(P)=a$ for all $P\in {H}_j$, then by the definition of $\calT_{m,n}$ there exists a permutation $\sigma$ over $\ma$ such that $\hist(\sigma(P))= \hist(P)$ and $\sigma(a)\ne a$. This means that $\sigma(P)\in {H}_j$,  and  therefore, $r(\sigma(P)) =a\ne \sigma(r(P))$, which means that $\rel_\text{neu}(P)=0$. Again, inequality (\ref{equ:anoneu}) holds for $P\in \hist^{-1}({H}_j)$. This proves the lemma.
\end{proof}

{\bf Step 2.~Smoothed possibility.}  First, we define $r_\text{ano}$ and $r_\text{neu}$ by extending a voting rule $r^*$ that is defined on $\ml(\ma)^n\setminus \calT_{m,n}$ and is guaranteed to satisfy anonymity and neutrality simultaneously for all profiles that are not in $\calT_{m,n}$. 

More precisely,   for any pair of profiles $P$ and $P'$, we write $P\equiv P'$ if and only if there exists a permutation $\sigma\in \ms_\ma$ such that $\hist(P') = \sigma(\hist(P))$. Let $(\ml(\ma)\setminus \calT_{m,n})=D_1\cup \cdots \cup D_{L'}$ denote the partition w.r.t.~$\equiv$. For each $l\le L'$,  let $P_l\in D_{l}$ denote an arbitrary profile. Because $P_l\not\in\calT_{m,n}$, there exists an alternative $a_l$ that is not covered by $\perm(P)$. Then, for any $P'\in D_l$ and any permutation $\sigma$ such that $\hist(P')=\sigma(\hist(P))$, we let $r^*(P')=\sigma(a_l)$. We note that the choice of $\sigma$ does not matter, because  all permutations $\sigma$ that map $\hist(P)$  to $\hist(P')$ have the same image for $a_l$. To see this, for the sake of contradiction, suppose $\hist(P')=\sigma_1(\hist(P)) = \sigma_2(\hist(P))$ where $\sigma_1(a_l)\ne \sigma_2(a_l)$. Then, we have $\sigma_1^{-1}\circ \sigma_2\in\perm(P)$, and $(\sigma_1^{-1}\circ \sigma_2)(a_l) \ne a_l$, which is a contradiction to the assumption that $\perm(P)$ does not cover $a_l$.

\begin{claim}\label{claim:rstar} For any profile $P\not\in \calT_{m,n}$ we have $\rel_\text{ano}(r^*,P)=1$ and $\rel_\text{neu}(r^*,P)=1$.
\end{claim}
\begin{proof} 
By definition, for all profiles $P\not\in \calT_{m,n}$, $r^*(P)$ only depends on $\hist(P)$, which means that $\rel_\text{ano}(r^*,P)=1$. To prove $\rel_\text{neu}(r^*,P)=1$, let $P\in D_l$ for some $l\le L$ and let $\sigma\in \ms_\ma$ denote a permutation over alternatives. We first prove that $\sigma(P)\in D_l$. Let $P' = \sigma(P)$. For the sake of contradiction suppose $P'\not\in D_l$. Then we have $P'\in \calT_{m,n}$. Therefore, for any alternative $a\in \ma$, there exists $\eta\in \perm(P')$ such that $\eta(\sigma(a))\ne \sigma(a)$. It follows that $(\sigma^{-1}\circ\eta\circ\sigma)\in \perm(P)$ and $(\sigma^{-1}\circ\eta\circ\sigma)(a)\ne a$. Therefore, $\perm(P)$ covers $\ma$, which contradicts the assumption that $P\not\in \calT_{m,n}$. Let $P = \zeta(P_l)$, where $P_l$ is the profile chosen in the definition of $r^*$. It follows that $P' = (\sigma\circ\zeta) (P_l)$. Therefore, we have $r^*(P') = (\sigma\circ\zeta)(a_l) =\sigma(\zeta(a_l)) =\sigma(r^*(P))$. This proves the claim. \end{proof}

For any profile $P\not\in\calT_{m,n}$, we let $r_\text{ano}(P)=r_\text{neu}(P) = r^*(P)$. For any profile $P$ in $\calT_{m,n}$, we let $r_\text{ano}(P)=a$ for an arbitrary fixed alternative $a$, and let $r_\text{neu}(P)$ be the top-ranked alternative of  agent $1$. By Claim~\ref{claim:rstar}, for any $r\in \{r_\text{ano}, r_\text{neu}\}$ any $n$, and any $\vec \pi\in \Pi^n$, $\Pr_{P\sim\vec\pi}(\rel_\text{ano}(r, P)+\rel_\text{neu}(r, P)<2)\ge \Pr_{P\sim\vec\pi}(P\in {\calT}_{m,n})$. Therefore, it suffices to prove the following lemma.

\begin{lem}\label{lem:smoothedposs}For any $r\in \{r_\text{ano}, r_\text{neu}\}$, any $n$, and any $\vec\pi\in \Theta^n$, we have:
$$\Pr\nolimits_{P\sim\vec\pi}(P\in {\calT}_{m,n})=\left\{\begin{array}{ll}O(n^{- \frac{l_\Pi}{2}})&\text{if } \calU_m^\Pi\ne\emptyset \\  \exp(-\Omega(n))&\text{otherwise}\end{array}\right..$$
\end{lem}
\begin{proof}
The lemma is proved by applying the upper bound in Lemma~\ref{lem:maintech} as in the proof of Theorem~\ref{thm:smoothedCondorcet}. For any $U\in \calU_m$, we define constrains $\pcons{U}$ as follows.
\begin{dfn}\label{dfn:consu}
For any $U\in \calU_m$, we define  $\pcons{U}$ as in Definition~\ref{dfn:cab} that is based on $\pba{U}$ and $\pbb{U}$, where $\pba{U}$ represents $\{x_{R}-x_{\sigma(R)}=0 :\forall R\in\ml(\ma),\forall \sigma\in U\}$ and $\pbb{U}=\emptyset$. Let $\ppoly{U}$ and $\ppolyz{U}$ denote the solutions to $\pcons{U}$ and its relaxation $\overline{\text{C}}_U$ as in Definition~\ref{dfn:cab}, respectively. 
\end{dfn}
By definition, $\pba{U}(x_R)=0$ if and only if for all $\sigma\in U$, $x_R = \sigma(x_R)$. Because $\pbb{U}=\emptyset$, we have $\pcons{U}=\overline{\text{C}}_U$, which means that  $\ppoly{U}=\ppolyz{U}$.
We have the following claim about $\pcons{U}$.
\begin{claim}\label{claim:anoneusol}
For any profile $P$  and any $U\in \calU_{m}$, $\hist(P)\in \ppoly{U}$ if and only if for all $\sigma\in U$, $\hist(P) = \sigma(\hist(P))$.  Moreover, $\rank(\pba{U})=(1-\frac{1}{|U|})m!$.
\end{claim}
\begin{proof} The  ``if and only if'' part follows after the definition. We now prove that $\rank(\pba{U})=(1-\frac{1}{|U|})m!$. Let $\equiv_U$ denote the relationship over $\ml(\ma)$ such that for any pair of linear orders $R_1, R_2$, $R_1\equiv_U R_2$ if and only if there exists $\sigma\in U$ such that $R_1=\sigma(R_2)$. Because $U$ is a permutation group, $\equiv_U$ is an equivalence relationship that partitions $\ml(\ma)$ into $\frac{m!}{|U|}$ groups, each of which has $|U|$ linear orders. It is not hard to see that each equivalent class is characterized by $|U|-1$ linearly independent equations represented by rows in $\pba{U}$, which means that $\rank(\pba{U})\le (|U|-1)(\frac{m!}{|U|})= (1-\frac{1}{|U|})m!$. For any $s< (1-\frac{1}{|U|})m!$ and any combination of $s$ rows of $\pba{U}$, denoted by $A$, it is not hard to construct $\vec x$ such that $A\cdot (\vec x)^\top=(\vec 0)^\top$ but  $\pba{U}\cdot (\vec x)^\top\ne (\vec 0)^\top$, which proves that $\rank(\pba{U})\ge (1-\frac{1}{|U|})m!$. This proves the claim.
\end{proof}

Let $\mc$ denote the set of all $\pcons{U}$ where $U\in \calU_m$ as in Definition~\ref{dfn:consu}. We have the following observations. 
\begin{itemize}
\item[(1)] $\mc$ characterizes $\calT_{m,n}$, because by Claim~\ref{claim:anoneusol}, for any $P\in \calT_{m,n}$ we have $\hist(P)\in \ppoly{\perm(P)}$ and $\perm(P)\in \mc$; and vice versa, for any $\pcons{U}\in \mc$ and any $P$ such that $\hist(P)\in \ppoly{U}$, by Claim~\ref{claim:anoneusol} we have $\perm(P)\supseteq U$, which means that $\perm(P)$ covers $\ma$, and it follows that $P\in \calT_{m,n}$. 
\item [(2)] For any $U\in \calU_m$, $\ppoly{U}\ne \emptyset$. This because $\vec 1\in \ppoly{U}$.  
\item [(3)] $|\mc|$ can be seen as a constant that does not depend on $n$, because it is no more than the total number of permutation groups over $\ma$. 
\end{itemize}

We now prove the polynomial upper bound when $\calU_m^\Pi\ne \emptyset$. For any $U\in \calU_m^\Pi$, by Claim~\ref{claim:anoneusol}, we have $\rank(\pba{U}) \ge  (1-\frac{1}{|U|})m! \ge  (1-\frac{1}{l_{\min}})m!$. By applying the polynomial upper bound in Lemma~\ref{lem:maintech} to all $\pcons{U}\in \mc$,  for any $\vec\pi\in\Pi^n$, we have:
\begin{align*}
&\Pr\nolimits_{P\sim\vec\pi}(P\in {\calT}_{m,n}) \le \sum_{\pcons{U}\in \mc} \Pr\nolimits_{P\sim\vec\pi}(\hist(P)\in \ppoly{U})\\
=& \sum_{\pcons{U}\in \mc} O \left(n^{- (\frac{|U|-1}{2|U|})m!}\right)
\le | \mc| O \left(n^{- (\frac{l_{\min}-1}{2l_{\min}})m!}\right)
=O\left(n^{- \frac{l_\Pi}{2}m!}\right)
\end{align*}
The exponential upper bound is proved similarly, by applying the exponential upper bound in Lemma~\ref{lem:maintech} to all $\pcons{U}\in \mc$. This proves the lemma.
\end{proof}

{\bf Step 3.~Smoothed impossibility}.  Lemma~\ref{lem:sumdegree} implies that for any $\vec\pi\in\Pi^n$, $\Pr_{P\sim\vec\pi}(\rel_\text{ano}(r, P)+\rel_\text{neu}(r, P)<2)\ge \Pr_{P\sim\vec\pi}(P\in {\calT}_{m,n})$. Therefore, it suffices to prove the following lemma.

\begin{lem}\label{lem:smoothedimp}For any voting rule $r$, there exist infinitely many $n\in\mathbb N$ and corresponding $\vec\pi\in\Pi^n$, such that:
$$\Pr\nolimits_{P\sim\vec\pi}(P\in {\calT}_{m,n})=\left\{\begin{array}{ll}\Omega(n^{- \frac{l_\Pi}{2}})&\text{if } \calU_m^\Pi\ne\emptyset \\  \exp(-O(n))&\text{otherwise}\end{array}\right..$$
\end{lem}
\begin{proof} Let $\mc$ be the set as defined in the proof of Lemma~\ref{lem:smoothedposs}. 

{\bf Applying the lower bound in Lemma~\ref{lem:maintech}.} We first prove the polynomial lower bound. Suppose $\calU_m^\Pi\ne \emptyset$ and let $U\in \calU_m^\Pi$ denote the permutation group with the minimum size, i.e.~$|U|=l_{\min}$. By Claim~\ref{claim:anoneusol}, $\rank(\pba{U}) = (1-\frac{1}{l_{\min}})m!$. By applying the tightness of the polynomial bound part in Lemma~\ref{lem:maintech} to $\pcons{U}$, we have that there exists constant $C>0$ such that for any $n'\in \mathbb N$, there exists $n'\le n\le Cn'$ and $\vec \pi\in \Pi^n$ such that 
$$\Pr\nolimits_{P\sim\vec\pi}(\hist(P)\in \ppoly{U})=\Omega\left(n^{- (\frac{l_{\min}-1}{2l_{\min}})m!}\right)$$ 
Note that for any profile $P$ such that $\hist(P)\in\ppoly{U}$, we have $\perm(P)\supseteq U$, which means that $\perm(P)$ covers $\ma$ and therefore $P\in \calT_{m,n}$. Therefore, we have:
\begin{align*}
\Pr\nolimits_{P\sim\vec\pi}(P\in {\calT}_{m,n}) \ge \Pr\nolimits_{P\sim\vec\pi}(\hist(P)\in\ppoly{U})\ge \Omega\left((Cn)^{- (\frac{l-1}{2l})m!}\right)
=\Omega\left(n^{- (\frac{l-1}{2l})m!}\right)
\end{align*} 

The exponential lower bound is proved similarly, by applying the tightness of the exponential bound part in Lemma~\ref{lem:maintech} to an arbitrary $U\in \calU_m$. This proves the lemma.
\end{proof}
This finishes the proof of Theorem~\ref{thm:smoothedANR}.
\end{proof}

\subsection{Connection to the (Non-)Existence of Anonymous and Neutral Voting Rules}
\label{sec:moulin}

As a side note, Lemma~\ref{lem:sumdegree} and Claim~\ref{claim:rstar} can be used to prove Moulin's characterization of existence of anonymous and neutral voting rules, which was stated as Problem 1 in~\cite{Moulin1983:The-Strategy} with hints on the proofs. More hints are given in~\cite[Problem~9.9]{Moulin91:Axioms}. We present a different proof using the group theoretic notation developed in this paper.
\begin{claim}[Problem 1 in~\cite{Moulin1983:The-Strategy}]\label{claim:anoneuexist} Fix any $m\ge 2$ and $n\ge 2$, there exists a voting rule that satisfies anonymity and neutrality if and only if $m$ cannot be written as the sum of $n$'s nontrivial divisors.
\end{claim}
\begin{proof} By Lemma~\ref{lem:sumdegree} and Claim~\ref{claim:rstar}, there exists a voting rule that satisfies anonymity and neutrality if and only if ${\calT}_{m,n}=\emptyset$. Therefore, it suffices to prove that ${\calT}_{m,n}=\emptyset$ if and only if $m$ cannot be written as the sum of $n$'s nontrivial divisors.

The ``if'' direction. Suppose for the sake of contradiction that ${\calT}_{m,n}\ne\emptyset$. Let $P\in {\calT}_{m,n}\ne\emptyset$ denote an arbitrary profile. We now partition $\ma$ according to the following equivalence relationship $\equiv$. $a\equiv b$ if and only if there exists $\sigma\in\perm(P)$ such that $\sigma(a) = b$. The partition is well defined because $\perm(P)$ is a permutation group. In other words, if $a\equiv b$ then $b\equiv a$, because the $\perm(P)$ is closed under inversion. If $a\equiv b$ (via $\sigma_1$) and $b\equiv c$ (via $\sigma_2$) then $a\equiv c$ because $c = (\sigma_2\circ\sigma_1)(a)$, and $\sigma_2\circ\sigma_1\in\perm(P)$.

Suppose $\ma=\ma_1\cup\cdots\cup \ma_L$ is divided into $L$ parts according to $\equiv$. Because $P\in{\cal T}_{m,n}$, which means that $\perm(P)$ covers $\ma$, for all $l\le L$, $|\ma_l|\ge 2$. It is not hard to check that for any $l\le L$, we have $|\ma_l|$ divides $|\perm(P)|$. Also it is not hard to see that $\perm(P)$ divides $n$. This means that $m=\sum_{l=1}^L|\ma_l|$ which contradicts the assumption. This proves the ``if'' direction.

The ``only if'' direction. Suppose for the sake of contradiction that $m=m_1+\cdots+m_L$ where each $m_l\ge 2$ and divides $n$. Consider the permutation $\sigma$ whose cycle form consists of $L$ cycles whose sizes are $m_1,\cdots,m_L$, respectively. Let $U$ denote the permutation group generated by $\sigma$. It follows that $|U|$ is the least common multiple of $m_1,\cdots,m_L$, which means that $|U|$ divides $n$. Therefore, it is not hard to construct an $n$-profile $P$ such that for any ranking $R$ and any $\sigma'\in U$, we must have $P[R]  = P[\sigma'(R)]$. It follows that $U\subseteq \perm(P)$, which means that $\perm(P)$ covers $\ma$, and therefore $\in{\cal T}_{m,n}\ne\emptyset$. This contradicts the assumption that $\in{\cal T}_{m,n}=\emptyset$, which proves the ``only if'' direction. 
\end{proof}

\section{Proof of Lemma~\ref{lem:taubounds}}
\label{appendix:lemmatauproof}

\appLem{lem:taubounds}{
For any $m\ge 2$, let  $l^* = \min_{U\in \calU_m}|U|$. We have  $l^* = 2$ if $m$ is even;  $l^* = 3$ if $m$ is odd and $3\mid m$; $l^* = 5$ if $m$ is odd, $3\nmid m$, and $5\mid m$; and $l^*=6$ for other $m$. 
}

\begin{proof} We prove the lemma by discussing the following cases.

{\bf Case 1: $\bm{2\mid m}$.} Let $\sigma=(1,2)(3,4)\cdots(m-1,m)$ denote the permutation that consists of 2-cycles. It follows that $\id = \sigma^2$ and $\{\id,\sigma\}$ is a permutation group that covers $[m]$. This means that $l^* = 2$.

{\bf Case 2: $\bm{2\nmid m}$ and $\bm{3\mid m}$.} Let $\sigma=(1,2,3)(4,5,6)\cdots(m-2,m-1,m)$ denote the permutation that consists of 3-cycles. It follows that $\id = \sigma^3$ and $\{\id,\sigma, \sigma^2\}$ is a permutation group that covers $[m]$. This means that $l^* \le 3$.  We now prove that $l^*$ cannot be $2$ by contradiction. Suppose for the sake of contradiction that $l^* =2$ and let $G$ denote a permutation group with $|G|=2$. Table~\ref{tab:groups5} (part of Table~26.1 in~\cite{Gallian2012:Contemporary}) lists all groups of orders $1$ through $5$ up to isomorphism. The order of a group is the number of its elements.  For any $l\in\mathbb N$, a permutation group $G$ of order $l$ is isomorphic to $Z_l$ if and only if $G=\{\id, \sigma, \ldots, \sigma^{l-1}\}$, where $\sigma$ is an order-$l$ permutation, that is, $\sigma^{l}=\id$ and for all $1\le i<l$, $\sigma^i\ne \id$. A permutation group $G$ of order $4$ is isomorphic to $Z_2\oplus Z_2$ if and only if $G=\{\id, \sigma, \eta, \sigma\circ\eta\}$, where $\sigma,\eta$, and $\sigma\circ\eta$ are three order-$2$ permutations.
\begin{table}[htp]
\caption{All order $1$ through $5$ groups up to isomorphism. \label{tab:groups5}}
\centering
\begin{tabular}{|c|c|c|c|c|c|}
\hline
Order& $1$& $2$& $3$& $4$& $5$ \\
\hline
Groups& $\{\id\}$& $Z_2$& $Z_3$& $Z_4$ or $Z_2\oplus Z_2$& $Z_5$\\
\hline
\end{tabular}

\end{table}

Therefore, $G = \{\id,\sigma\}$. Because $G$ covers $\ma$, the cycle form of $\sigma$ must consist of 2-cycles that cover all alternatives in $\ma$, which means that $2\mid m$, a contradiction.


{\bf Case 3: $\bm{2\nmid m}$, $\bm{3\nmid m}$, and $\bm{5\mid m}$.}  Let $\sigma=(1,2,3,4,5)\cdots(m-4,m-3,m-2,m-1,m)$ denote the permutation that consists of 5-cycles. It follows that $\id = \sigma^5$ and $\{\id,\sigma, \sigma^2, \sigma^3, \sigma^4\}$ is a permutation group that covers $[m]$. This means that $l^* \le 5$.  We now prove by contradiction that $l^*$ cannot be $2$, $3$, or $4$. Suppose for the sake of contradiction that $l^* =2$ and let $G$ denote a permutation group whose order is no more than $4$.  Because $2\nmid m$ and $3\nmid m$, following a similar argument with Case $2$ we know that $G$ cannot be isomorphic to $Z_2,Z_3$, or $Z_4$. Therefore, by Table~\ref{tab:groups5}, $G$ must be isomorphic to $Z_2\oplus Z_2$. However, the next claim shows that this is impossible. Therefore, $l^*=5$.

\begin{claim}\label{claim:z2z2} For any $m$ with $2\nmid m$, no permutation group in $G\in\gm$ is isomorphic to $Z_2\oplus Z_2$, where $\gm$ is the set of all permutation groups over $\ma = [m]$.
\end{claim}
\begin{proof} We prove the claim by contradiction. Suppose for the sake of contradiction $2\nmid m$ and there exists $G\in \gm$ that is isomorphic to $Z_2\oplus Z_2$. This means that
$G=\{\id, \sigma, \eta, \sigma\circ\eta\}$, where $\sigma$, $\eta$, and $\sigma\circ\eta$ only contain $2$-cycles in their cycle forms, respectively. Because $2\nmid m$, at least one alternative is not involved in any cycle in $\eta$. W.l.o.g.~we let $\{1,\ldots,k\}$ denote the alternatives that are not involved in any cycle in $\eta$, which means that they are mapped to themselves by $\eta$. It follows that the remaining $m-k$ alternatives are covered by $2$-cycles in $\eta$, which means that $k$ is an odd number. For any $i\le k$, because $G$ covers $\ma$, $i$ must be involved in a $2$-cycle in $\sigma$, otherwise none of $\eta,\sigma$, or $\sigma\circ \eta$ will map $i$ to a different alternative. Suppose $(i,j)$ is a $2$-cycle in $\sigma$. We note that $(\sigma\circ\eta)(i) = j$, which means that $(i,j)$ is a $2$-cycle in $\sigma\circ\eta$. This means that $\eta(j) = j$, that is, $j\le k$. Therefore, $\{1,\ldots, k\}$ consists of $2$-cycles in $\sigma$. This  means that $2\mid k$, which is a contradiction. 
\end{proof}

{\bf Case 4:  $\bm{2\nmid m}$, $\bm{3\nmid m}$, and $\bm{5\nmid m}$.}  Let $\sigma=(1,2,3)(4,5)\cdots(m-1,m)$. It follows that the order of $\sigma$ is $6$, which means that $\{\id,\sigma, \sigma^2, \sigma^3, \sigma^4, \sigma^5\}$ is a permutation group that covers $[m]$. This means that $l^* \le 6$. The proof for $l^*\ge 6$ is similar with the proof in Case 3. Suppose for the sake of contradiction there exists $G\in\gm$ with $|G|\le 5$. Then, by Table~\ref{tab:groups5}, $G$ must be isomorphic to $Z_2,Z_3,Z_4,Z_5$, or $Z_2\oplus Z_2$. However, this is impossible because none of $2$, $3$, or $5$ can divide $m$, and by claim~\ref{claim:z2z2}, $G$ is not isomorphic to   $Z_2\oplus Z_2$.  This proves that $l^* =6$.
\end{proof}
\subsection{A Corollary of Theorem~\ref{thm:smoothedANR} and Lemma~\ref{lem:taubounds}}\label{appendix:ANRcoro}

\begin{coro}\label{coro:smoothedANR}
Let $\mm= (\Theta,\ml(\ma),\Pi)$ be a strictly positive and closed single-agent preference model with $\pi_\text{uni}\in \conv(\Pi)$.  Let $l^*$ be the number in Lemma~\ref{lem:taubounds}. 

{\bf Smoothed possibility.} Let $r\in \{r_\text{ano}, r_\text{neu}\}$. For any $n$ and any $\vec\pi\in \Pi^n$, we have:
$$\Pr\nolimits_{P\sim\vec\pi}(\rel_\text{ano}(r, P)+\rel_\text{neu}(r, P)<2)=O\left(n^{- (\frac{l^*-1}{2l^*})m!}\right)$$

{\bf Smoothed impossibility.}  For any voting rule $r$, there exist infinitely many $n\in\mathbb N$ such that:
$$\sup\nolimits_{\vec\pi\in\Pi^n}\Pr\nolimits_{P\sim\vec\pi}(\rel_\text{ano}(r, P)+\rel_\text{neu}(r, P)<2)=\Omega\left(n^{- (\frac{l^*-1}{2l^*})m!}\right)$$

\end{coro}
In particular, Corollary~\ref{coro:smoothedANR} applies to all neutral, strictly positive, and closed models, which contains $\mm_\mallows^{[\underline{\varphi},1]}$, $\mm_\pl^{[\underline{\varphi},1]}$, and IC as  special cases. This is because for any neutral model, let $\pi\in\Pi$ denote an arbitrary distribution. Then, for any permutation $\sigma$ over $\ma$, $\sigma(\pi)\in \Pi$. Therefore, $\pi_\text{uni}=\frac{1}{m!}\sum_{\sigma\in\ms_\ma}\sigma(\pi)\in\conv(\Pi)$.

\section{Proof of Proposition~\ref{prop:nonoptlexfa}}
\label{sec:nonoptlexfaproof}

\appProp{prop:nonoptlexfa}{Let $r$  be a voting rule obtained from a positional scoring correspondence by applying  \lex{} or \fa{}. Let $\mm= (\Theta,\ml(\ma),\Pi)$ be a strictly positive and closed single-agent preference model with $\pi_\text{uni}\in \conv(\Pi)$. There exist infinitely many $n\in\mathbb N$ such that:

$\hfill\sup_{\vec\pi\in\Pi^n}\Pr_{P\sim\vec\pi}\left(\rel_\text{ano}(r, P)+\rel_\text{neu}(r, P)<2\right)=\Omega(n^{- 0.5})\hfill$
}

\begin{proof}
Suppose $r$ is obtained from a correspondence $\cor$ by applying   \lex{} or \fa{}. We first prove that any profiles $P$ with $|\cor(P)|=2$ violates anonymity or neutrality under $r$. In other words,  $\rel_\text{ano}(r, P)+\rel_\text{neu}(r, P)<2$. W.l.o.g.~suppose $c(P)=\{1,2\}$. Suppose $r$ is obtained from $\cor$ by applying \lex{}-$R$. W.l.o.g.~suppose $1\succ_R 2$. Let $\sigma$ denote the permutation that exchanges $1$ and $2$. Because $\cor$ is neutral, $\cor(\sigma(P))=\sigma(\cor(P))=\{1,2\}$. By \lex{}-$R$, $r(P)=r(\sigma(P))=1\ne 2 = \sigma(r(P))$, which violates neutrality. Suppose $r$ is obtained from $\cor$ by applying  \fa{}-$j$. W.l.o.g.~suppose $1\succ 2$ in the $j$-th vote in $P$. Because $\cor(P)=\{1,2\}$, there exists a vote in $P$ where $2\succ 1$. Suppose $2\succ 1$ in the $j'$-th vote. Let $P'$ denote the profile obtained from $P$ by switching $j$-th and $j'$-th vote. Because $c$ is anonymous, we have $\cor(P')=\{1,2\}$. By \fa{}-$j$, $r(P')=2\ne 1 = r(P)$, which violates anonymity.

Therefore, for any $n$ and any $\vec\pi\in\Pi^n$, we have $\Pr_{P\sim\vec\pi}\left(\rel_\text{ano}(r, P)+\rel_\text{neu}(r, P)<2\right)\ge \Pr_{P\sim\vec\pi}(|\cor(P)|=2)$. 
By applying the tightness of polynomial bound part in Lemma~\ref{lem:maintech}, it is not hard to prove that for any positional scoring correspondence $\cor$, there exist infinitely many $n\in\mathbb N$ and corresponding $\vec\pi\in\Pi^n$ such that  $\Pr_{P\sim\vec\pi}(|\cor(P)|=2)=\Omega(n^{-0.5})$. This proves the proposition.
\end{proof}

\section{Proof of Theorem~\ref{thm:mpsropt}}\label{sec:mpsroptproof}

\appThm{thm:mpsropt}{Let $\mm= (\Theta,\ml(\ma),\Pi)$ be a strictly positive and closed single-agent preference model with $\pi_\text{uni}\in \conv(\Pi)$. For any voting correspondence $\cor$ that satisfies anonymity and neutrality, let $r_\mpsr$ denote the voting rule obtained from $\cor$ by \mpsr{}-then-$\tb$. For any $n$ and any $\vec\pi\in \Pi^n$, we have:

$\hfill\Pr_{P\sim\vec\pi}(\rel_\text{ano}(r_\mpsr, P)+\rel_\text{neu}(r_\mpsr, P)<2)=O(n^{- \frac{m!}{4}})\hfill$

Moreover, if $\tb$ satisfies anonymity (respectively, neutrality) then $r_\mpsr$ also satisfies anonymity (respectively, neutrality).
}
\begin{proof} The upper bound is proved in the following two steps. First, we show that for any profile $P$ such that $\mpsr(P)\ne\emptyset$, 
\begin{equation}\label{equ:mpsranoneu}
\rel_\text{ano}(r_\mpsr,P)+\rel_\text{neu}(r_\mpsr,P) = 2
\end{equation} 
If  $|\cor(P)|=1$, then we have $r_\mpsr (P)=  \cor(P)$. (\ref{equ:mpsranoneu}) holds because $\cor$ satisfies neutrality and anonymity. If $\mpsr(P)\ne\emptyset$, then for any permutation $\sigma$ over $\ma$, we have $\cor(\sigma(P)) = \sigma(\cor(P))$ and it is not hard to see that  $\mpsr(\sigma(P))=\sigma(\mpsr(P))$. Let $\mpsr(P) = V$. Therefore, $r_\mpsr(\sigma(P))$ is the alternative in $\sigma(c(P))$ that is ranked highest in $\sigma(V)$, which means that $r_\mpsr(\sigma(P))=\sigma(r_\mpsr(P))$. This means that $\rel_\text{neu}(r_\mpsr,P)=1$. For any profile $P'$ with $\hist(P')=\hist(P)$, it is not hard to see that $\mpsr(P')=\mpsr(P)$. Because $\cor$ satisfies anonymity, we have $\cor(P') = \cor(P)$. This means that $\rel_\text{ano}(r_\mpsr,P) =1$.

Second, we show that $\Pr_{P\sim\vec\pi}(\mpsr(P)=\emptyset)=O(n^{- \frac{m!}{4}})$ by applying the polynomial upper bound in Lemma~\ref{lem:maintech} in a way similar to the proof of Theorem~\ref{thm:smoothedANR}. For any partition $\calQ=\{Q_1,\ldots,Q_L\}$ of $\ml(\ma)$, we define $\pcons{\calQ}$ as in Definition~\ref{dfn:cab} based on $\pba{\calQ}$, which represents $\{x_R-x_{R'}=0:\forall l\le L,\forall R,R'\in Q_l\}$, and $\pbb{\calQ}=\emptyset$. We have $\rank(\pba{\calQ})=m!-|\calQ|$. Let $\mc$ denote the set of all $\pcons{\calQ}$, where each set in $\calQ$ contains at least two linear orders. It is not hard to verify that $\mc$ characterizes all $P$ with $\mpsr(P)=\emptyset$, for any $\pcons{\calQ}\in \mc$, $\ppoly{\calQ}\ne\emptyset$ and $\rank(\pba{\calQ})\ge \frac{m!}{2}$, and $|\mc|$ does not depend on $n$. Because for any $\pcons{\calQ}\in \mc$, $\pi_\text{uni}\in \ppolyz{\calQ}\cap \conv(\Pi)$, we can apply the polynomial upper bound in  Lemma~\ref{lem:maintech} to all $\pcons{\calQ}\in \mc$, which gives us $\Pr_{P\sim\vec\pi}(\mpsr(P)=\emptyset)=O(n^{- \frac{m!}{4}})$.

The ``moreover'' part follows after noticing that (1) for any pair of profiles $P$ and $P'$ such that $\mpsr(P)=\emptyset$ and $\hist(P')=\hist(P)$, we have $\mpsr(P')=\emptyset$, and (2) for any profile $P$ with $\mpsr(P)=\emptyset$ and any permutation $\sigma$ over $\ma$, we have $\mpsr(\sigma(P))=\emptyset$.
\end{proof}

\section{Smoothed Likelihood of Other Commonly-Studied Events in Social Choice}\label{sec:application}
In this section we show how to apply Lemma~\ref{lem:maintech} to obtain dichotomy results on smoothed likelihood of various social choice events. Let us start with a dichotomy result on the non-existence of Condorcet cycles. 


\begin{prop}[\bf Smoothed likelihood of non-existence of Condorcet cycles]\label{prop:nowcwcycle} Let $\mm= (\Theta,\ml(\ma),\Pi)$ be a strictly positive and closed single-agent preference model. 

\text{\bf Upper bound.} For any $n\in \mathbb N$ and any $\vec\pi\in \Pi^n$, we have:
$$
\Pr\nolimits_{P\sim\vec\pi}(\ncc(P)=1)=\left\{\begin{array}{ll}O(1)&\text{if  }\exists  \pi\in \conv(\Pi)\text{ s.t. }\umg(\pi)\text{ is acyclic} 
\\ \exp(-\Omega(n))&\text{otherwise}\end{array}\right..$$

\text{\bf Tightness of the upper bound.} There exist infinitely many $n\in \mathbb N$ such that:
$$
\sup\nolimits_{\vec\pi\in\Pi^N}\Pr\nolimits_{P\sim\vec\pi}(\ncc(P)=1)=\left\{\begin{array}{ll}\Omega(1)&\text{if }\exists  \pi\in \conv(\Pi)\text{ s.t. }\umg(\pi)\text{ is acyclic} 
\\ \exp(-O(n))&\text{otherwise}\end{array}\right..
$$

\end{prop}
\begin{proof}  
The proof proceeds in the following three steps.

{\bf First step: defining $\bm\mc$.}   Let $\mc$ denote the set of all $\pcons{G}$ where $G$ is an acyclic unweighted directed graph, as in Definition~\ref{dfn:cg}.  We have the following observations. 
\begin{itemize}
\item [(1)] For any profile $P$, $\ncc(P)=1$ if and only if $\hist(P)\in\bigcup_{\pcons{G}\in\mc}\ppoly{G}$. To see this, if $\ncc(P)=1$ then $\umg(P)$ is acyclic, which means that $\hist(P)\in \ppoly{\umg(P)}$, where $\pcons{\umg(P)}\in\mc$; and conversely, if $\hist(P)\in \ppoly{G}$ for an acyclic graph $G$, then $\ncc(P)=1$. 
\item [(2)] For any graph $G$, $\ppoly{G}\ne\emptyset$ by McGarvey's theorem~\cite{McGarvey53:Theorem}. 
\item [(3)] $|\mc|$ only depends on  $m$, which means that $|\mc|$ can be seen as a constant that does not depend on $n$. 
\end{itemize}

{\bf \boldmath  Second step, the $O(1)$ case.}  The $O(1)$ upper bound is straightforward.  To prove its tightness, suppose there exists $\pi\in\conv(\Pi)$ such that $\umg(\pi)$ is acyclic. This means that there exists a topological ordering of $\umg(\pi)$. Let $G$ denote an arbitrary complete acyclic supergraph of $\umg(\pi)$. It follows that $\pcons{G}\in\mc$, and $\pi\in\ppolyz{G}$ due to Claim~\ref{claim:solumg}. Therefore, $\pcons{G}\cap \conv(\Pi)\ne \emptyset$. Also note that $\pba{G}=\emptyset$, which means that $\rank(\pba{G})=0$. The tightness follows after applying the tightness of the polynomial part in Lemma~\ref{lem:maintech} to $\pcons{G}$.

{\bf Third step, the exponential case.}  For any $\pcons{G}\in\mc$ and any $\pi\in\conv(\Pi)$, we first show that $\ppolyz{G}\cap \conv(\Pi)=\emptyset$. Suppose for the sake of contradiction there exists an acyclic graph $G$ such that $\pi\in {\ppolyz{G}}\cap \conv(\Pi)$. Then, by Claim~\ref{claim:solumg}, $\umg(\pi)$ is a subgraph of $G$, which means that $\umg(\pi)$ is acyclic. This contradicts the assumption on $\conv(\Pi)$ in the exponential case.  The upper bound (respectively, its tightness)  follows after applying the exponential upper bound (respectively, its tightness) in Lemma~\ref{lem:maintech} to all $\pcons{G}\in \mc$ (respectively, an arbitrary $\pcons{G}\in \mc$).
\end{proof}

The dichotomy results in this section will be presented by the following template exemplified by Proposition~\ref{prop:nowcwcycle}. In the template, we will specify three components: 
\begin{itemize}
\item {\sc event}, which is an event of interest that depends on the profile $P$, 
\item {\sc condition}, which is often about the existence of $\pi\in\conv(\Pi)$ that satisfies a weaker version of  {\sc event},  and 
\item a number $l_\Pi$ that depends on the statistical model $\mm$.
\end{itemize}

{\sc \bf Template for dichotomy results} {\bf(Smoothed likelihood of {\sc event})}. {\em Let $\mm= (\Theta,\ml(\ma),\Pi)$ be a strictly positive and closed single-agent preference model.
  
\text{\bf Upper bound.} For any $n\in \mathbb N$ and any $\vec\pi\in \Pi^n$, we have:
$$
\Pr\nolimits_{P\sim\vec\pi}(\text{\sc event})=\left\{\begin{array}{ll}O(n^{-\frac{l_\Pi}{2}})&\text{if {\sc condition} holds}\\ \exp(-\Omega(n))&\text{otherwise}\end{array}\right..$$

\text{\bf Tightness of the upper bound.} There exist infinitely many $n\in \mathbb N$ such that:
$$
\sup\nolimits_{\vec\pi\in\Pi^N}\Pr\nolimits_{P\sim\vec\pi}(\text{\sc event})=\left\{\begin{array}{ll}\Omega(n^{-\frac{l_\Pi}{2}})&\text{if {\sc condition} holds}\\ \exp(-O(n))&\text{otherwise}\end{array}\right..
$$
}

For example, in Proposition~\ref{prop:nowcwcycle}, {\sc event} is ``there is no Condorcet cycle'', {\sc condition} is ``there exists $\pi\in\conv(\Pi)$ such that $\umg(\pi)$ is acyclic'', and $l_\Pi = 0$. Table~\ref{tab:propositions} summarizes the dichotomy results using the template. A {\em Condorcet winner} is the alternative who beats every other alternative in their head-to-head competition.  If a Condorcet winner exists, then it must be unique. A {\em weak Condorcet winner} is an alternative who never loses in head-to-head competitions. Weak Condorcet winners may not be unique.


\newcommand{\insertRow}[5]{
\parbox[c][#1]{.05\textwidth}{\centering #2} & 
\parbox[c][#1]{.26\textwidth}{ #3} & 
\parbox[c][#1]{.47\textwidth}{\raggedright #4} & 
\parbox[c][#1]{.12\textwidth}{\centering #5} \\
\hline
}
\begin{table}[htp]
\centering
\caption{Summary of dichotomy results on smoothed likelihood of events. $\conv(\Pi)$ is the convex hull of $\Pi$. $\umg$ is the unweighted majority graph.\label{tab:propositions}}
\begin{tabular}{|p{.05\textwidth}|p{.26\textwidth}|p{.47\textwidth}|p{.12\textwidth}|}
\hline
\centering Prop.& 
\centering {\sc event} &\centering {\sc condition} & \parbox[c]{.12\textwidth}{\centering $l_\Pi$} \\
\hline
\insertRow{1cm}{\ref{prop:nowcwcycle}}{ No Condorcet cycles }{ $\exists \pi\in\conv(\Pi)$ s.t.~$\umg(\pi)$ is acyclic}{ $0$}
\insertRow{1cm}{\ref{prop:wcwcycle}}{ $\exists$ Condorcet cycle of length $k$ }{ $\exists \pi\in\conv(\Pi)$ s.t.~$\umg(\pi)$  contains a weak Condorcet cycle of length $k$}{ $0$}
\insertRow{1cm}{\ref{prop:cw}}{$\exists$ Condorcet winner}{ $\exists \pi\in\conv(\Pi)$ that has at least one weak Condorcet winner }{ $0$}
\insertRow{1cm}{\ref{prop:nocw}}{ No Condorcet winner }{ $\exists \pi\in\conv(\Pi)$ and a supergraph $G$ of $\umg(\pi)$ that has no weak Condorcet winner }{ $0$ or  $1$}
\insertRow{1cm}{\ref{prop:wcw}}{$\exists$ exactly $k$ weak Condorcet winners}{$\exists \pi\in\conv(\Pi)$ that contains at least $k$ weak Condorcet winners}{$\dfrac{k(k-1)}{2}$}\insertRow{1cm}{\ref{prop:nowcw}}{No weak Condorcet winners }{$\exists \pi\in\conv(\Pi)$ and a supergraph $G$ of $\umg(\pi)$ that has no weak Condorcet winner }{$0$}

\end{tabular}
\end{table}

Any dichotomy result using the template is quite general, because it applies to all  strictly positive and closed single-agent preference models, any $n$, and any combination of distributions. In particular, it is not hard to verify that if $\pi_\text{uni}\in \conv(\Pi)$, where $\pi_\text{uni}$ is the uniform distribution over $\ml(\ma)$, then {\sc condition} is satisfied for all propositions in Table~\ref{tab:propositions}, which means that the polynomial bounds apply. For example, because $\pi_\text{uni}\in \conv(\Pi)$ for any neutral model (for any $\pi\in \Pi$, the average of $m!$ distributions obtained from $\pi$ by applying all permutations is $\pi_\text{uni}$), we have the following corollary.

\begin{coro}\label{coro:table} The polynomial bounds in Table~\ref{tab:propositions} apply to all neutral,  strictly positive, and closed models including $\mm_\mallows^{[\underline{\varphi},1]}$ and $\mm_\pl^{[\underline{\varphi},1]}$ for all $0<\underline\varphi\le 1$, and IC, which corresponds to $\Pi=\{\pi_\text{uni}\}$.
\end{coro}
In light of Corollary~\ref{coro:table} and as a result of Proposition~\ref{prop:wcwcycle}, the likelihood of Condorcet voting paradox is asymptotically maximized under IC, among all i.i.d.~distributions over $\ml(\ma)$. This givens an asymptotic answer to an open questions by~\citet{Tsetlin2003:The-impartial}.

If {\sc event} is desirable, such as ``Condorcet winner'' or ``No Condorcet cycles'', then a polynomial (sometimes $\Theta(1)$) smoothed likelihood is desirable; if {\sc event} is undesirable, such as ``No Condorcet winner'' or ``there exists a Condorcet cycle of length $k$'', then an exponential likelihood is desirable. 

\paragraph{\bf Overview of proof techniques.}  All propositions are proved by applying Lemma~\ref{lem:maintech} in the following three steps exemplified by the proof of Proposition~\ref{prop:nowcwcycle}. {\bf First, we define a set $\mc$} of constraints $\cons$'s  such that
\begin{itemize}
\item[(1)] {\sc event} is characterized by $\bigcup_{\cons\in\mc}\sol$ in the sense that {\sc event}  holds for a profile $P$ if and only if $\hist(P)\in \sol$ for some $\cons\in\mc$, 
\item[(2)] for each $\cons\in\mc$, $\sol\ne\emptyset$, and 
\item [(3)] $|\mc|$ is a constant that does not depend on $n$ (but may depend on $m$). 
\end{itemize}
{\bf Second, for  the polynomial bound}, the upper bound is proved by applying the polynomial upper bound in Lemma~\ref{lem:maintech} to all (constant number of) $\cons\in\mc$.  The tightness is proved by explicitly choosing $\cons\in\cal C$, often as a function of some $\pi\in\conv(\Pi)$, so that $\pi\in \rsol$, which implies $\rsol\cap \conv(\Pi)\ne\emptyset$, and then applying the tightness of the polynomial bound in Lemma~\ref{lem:maintech}.   {\bf Third, for the  exponential bound}, we first prove that for any $\cons\in\mc$, $\rsol\cap \conv(\Pi)=\emptyset$. Then, the upper bound (respectively, its tightness) is proved by applying the exponential bound (respectively, its tightness) in Lemma~\ref{lem:maintech} to all $\cons\in \mc$ (respectively, an arbitrary $\cons\in \mc$).


\begin{dfn}
For any $3\le k\le m$, we let $\rel_{\text{CC}=k}(P)=1$ (respectively, $\rel_{\text{WCC}=k}(P)=1$) if there exists a  Condorcet cycle (respectively, weak Condorcet cycle) of length $k$ in $P$; otherwise $\rel_{\text{CC}=k}(P)=0$ (respectively, $\rel_{\text{WCC}=k}(P)=0$).
\end{dfn}
\begin{prop}[\bf Smoothed likelihood of existence of Condorcet cycles with length $k$]\label{prop:wcwcycle} 
Let $\mm= (\Theta,\ml(\ma),\Pi)$ be a strictly positive and closed single-agent preference model.
  
\text{\bf Upper bound.} For any $n\in \mathbb N$, any $\vec\pi\in \Pi^n$, and any $3\le k\le m$, we have:
$$
\Pr\nolimits_{P\sim\vec\pi}(\rel_{\text{CC}=k}(P))=\left\{\begin{array}{ll}O(1)&\text{if }\exists \pi\in \conv(\Pi)\text{ s.t. }\rel_{\text{WCC}=k}(\pi)=1\\ \exp(-\Omega(n))&\text{otherwise}\end{array}\right..$$

\text{\bf Tightness of the upper bound.} For any $3\le k\le m$, there exist infinitely many $n\in \mathbb N$ such that:
$$
\sup\nolimits_{\vec\pi\in\Pi^N}\Pr\nolimits_{P\sim\vec\pi}(\rel_{\text{CC}=k}(P))=\left\{\begin{array}{ll}\Omega(1)&\text{if  }\exists \pi\in \conv(\Pi)\text{ s.t. }\rel_{\text{WCC}=k}(\pi)=1\\ \exp(-O(n))&\text{otherwise}\end{array}\right..
$$
\end{prop}
\begin{proof} 
 {\bf First step, defining $\bm\mc$}. For any length-$k$ cycle $p=a_1\ra a_2\ra\cdots\ra a_k\ra a_1$ in $\ma$, we define  $\pcons{p}$, where $\pba{p} = \emptyset$ and $\pbb{p}$ represents the $k$ constraints 
 $$\{\pair_{a_2,a_1}(\vec x_\ma)<0, \pair_{a_3,a_2}(\vec x_\ma)<0, \ldots,\pair_{a_1,a_k}(\vec x_\ma)<0\}$$
  Let $\mc$ denote all such $\pcons{p}$. We have the following observations. 
\begin{itemize}
\item [(1)] For any profile $P$, $\rel_{\text{CC}=k}(P)=1$ if and only if $\hist(P)\in\bigcup_{\pcons{p}\in\mc}\ppoly{p}$. To see this,  if $\rel_{\text{CC}=k}(P)=1$ then there exists a length-$k$ cycle $p$ in $\umg(P)$, which means that $\hist(P)\in \ppoly{p}$, where $\pcons{p}\in\mc$; and conversely, if $\hist(P)\in \ppoly{p}$ for some length-$k$ cycle $p$, then $p$ is a length-$k$ Condorcet cycle in $P$, which means that $\rel_{\text{CC}=k}(P)=1$. 
\item[(2)] by McGarvey's theorem~\cite{McGarvey53:Theorem}, for each $\pcons{p}\in\mc$, there exists a profile $P$ where $p$ is a cycle in $\umg(P)$, which means that $\ppoly{p}\ne\emptyset$. 
\item [(3)] The total number of length-$k$ cycles in $\ma$ only depends on $m$ and $k$, which means that $|\mc|$ can be seen as a constant that does not depend on $n$. 
\end{itemize}

{\bf Second step, the polynomial case.} The $O(1)$ upper bound is straightforward.  To prove the tightness, suppose there exists $\pi\in\conv(\Pi)$ with $\rel_{\text{WCC}=k}(\pi)=1$. Let $p$ denote an arbitrary length-$k$ weak Condorcet cycle in $\umg(\pi)$. It follows that $\pbb{p}\cdot (\pi)^\top\le \invert{\vec 0}$, which means that ${\ppolyz{p}}\cap \conv(\Pi)\ne\emptyset$, because $\pba{p}=\emptyset$. The tightness follows after applying the tightness of the polynomial lower bound in Lemma~\ref{lem:maintech} to $\pcons{p}$, where $\rank(\pba{p})=0$.

{\bf Third step, the exponential case.}  For any $\pcons{p}\in\mc$ and any $\pi\in\conv(\Pi)$, we first prove that $\ppolyz{p}\cap \conv(\Pi)=\emptyset$. Suppose for the sake of contradiction that there exists a length-$k$ cycle $p$ such that $\pi\in {\ppolyz{p}}\cap \conv(\Pi)$. Then, because $\pbb{p}\cdot (\pi)^\top\le \invert{\vec 0}$, for any edge $a_i\ra a_{i+1}$ in $p$ we must have $\pi[a_{i},a_{i+1}]-\pi[a_{i+1},a_{i}]\ge 0$, which means that $p$ is a length-$k$ weak Condorcet cycle in $\umg(\pi)$, meaning that $\rel_{\text{WCC}=k}(\pi)=1$. This contradicts the assumption on $\conv(\Pi)$ in the exponential case.  The upper bound (respectively, its tightness) follows after applying the exponential bound (respectively, its tightness) in Lemma~\ref{lem:maintech} to all (respectively, an arbitrary) $\pcons{p}\in \mc$.
\end{proof}

\begin{dfn}
For any profile $P$, let $\cw(P)=1$ if $P$ has a Condorcet winner; otherwise let $\cw(P)=0$.  Let $\wcw(P)$ denote the number of weak Condorcet winners in $P$.
\end{dfn}
\begin{prop}[\bf Smoothed likelihood of existence of Condorcet winner]\label{prop:cw} Let $\mm= (\Theta,\ml(\ma),\Pi)$ be a strictly positive and closed single-agent preference model.
  
\text{\bf Upper bound.} For any $n\in \mathbb N$ and any $\vec\pi\in \Pi^n$, we have:
$$
\Pr\nolimits_{P\sim\vec\pi}(\cw(P)=1)=\left\{\begin{array}{ll}O(1)&\text{if }\exists\pi\in \conv(\Pi)\text{ s.t. }\wcw(\pi)\ge 1\\ \exp(-\Omega(n))&\text{otherwise}\end{array}\right..$$

\text{\bf Tightness of the upper bound.} There exist infinitely many $n\in \mathbb N$ such that:
$$
\sup\nolimits_{\vec\pi\in\Pi^n}\Pr\nolimits_{P\sim\vec\pi}(\cw(P)=1)=\left\{\begin{array}{ll}\Omega(1)&\text{if }\exists\pi\in \conv(\Pi)\text{ s.t. }\wcw(\pi)\ge 1\\ \exp(-O(n))&\text{otherwise}\end{array}\right..
$$
\end{prop} 
\begin{proof} {\bf First step: defining $\bm\mc$}.   Let  $\mc=\{\pcons{G}: \cw(G) = 1\}$, that is,  $\mc$ contains all $\pcons{G}$ (Definition~\ref{dfn:cg}) where $G$ contains a Condorcet winner. We have the following observations. 
\begin{itemize}
\item [(1)] For any profile $P$, $\cw(P)=1$ if and only if $\hist(P)\in\bigcup_{\pcons{G}\in\mc}\ppoly{G}$. To see this, if $\cw(P)=1$ then $\hist(P)\in\ppoly{\umg(P)}\in \mc$, where $\pcons{G}\in\mc$; and conversely, if $\hist(P)\in \ppoly{G}$ for some $\pcons{G}\in\mc$, then $P$ has a Condorcet winner. 
\item [(2)] For any $G$ with $\cw(G) = 1$,  $\ppoly{G}\ne\emptyset$ due to McGarvey's theorem~\cite{McGarvey53:Theorem}. 
\item [(3)] $|\mc|$ can be seen as a constant that does not depend on $n$. 
\end{itemize}

{\bf Second step: the polynomial case.} The $O(1)$ upper bound trivially holds. To prove the tightness, suppose there exits $\pi\in \conv(\Pi)$ such that $\wcw(\pi)\ge 1$. Let $a$ denote an arbitrary weak Condorcet winner in $\umg(\pi)$. We obtain a complete graph $G^*$ from $\umg(\pi)$ by adding $a\ra b$ for all tied pairs $(a,b)$ in $\umg(\pi)$, and then adding arbitrary edges between other tied pairs in $\umg(\pi)$. It follows that $\cw(G^*)=1$ and $\pba{G^*}=\emptyset$ because $G^*$ is complete, which means that $\rank(\pba{G^*})=0$. The tightness follows after applying the tightness of the polynomial bound  in Lemma~\ref{lem:maintech} to $\pcons{G^*}$.

{\bf Third step: the exponential case.}  For any $\pcons{G}\in\mc$ and any $\pi\in\conv(\Pi)$, we now prove that $\ppolyz{G}\cap \conv(\Pi)=\emptyset$. Suppose for the sake of contradiction there exist $\pcons{G}\in\mc$ and $\pi\in {\ppolyz{G}}\cap \conv(\Pi)$. By Claim~\ref{claim:solumg}, $\umg(\pi)$ is a subgraph of $G$, which means that the Condorcet winner in $G$ is a weak Condorcet winner in $\umg(\pi)$, which contradicts the assumption that  $\wcw(\pi)=0$ in the exponential case. The upper bound  (respectively, its tightness) follows after applying the exponential bound (respectively, its tightness) in Lemma~\ref{lem:maintech} to all (respectively, an arbitrary) $\pcons{G}\in \mc$.
\end{proof}

\begin{prop}[\bf Smoothed likelihood of non-existence of Condorcet winner]\label{prop:nocw}  Let $\mm= (\Theta,\ml(\ma),\Pi)$ be a strictly positive and closed single-agent preference model. Let $\mg_\Pi$ denote the set of all unweighted directed graphs $G$ over $\ma$ such that (1) there exists $\pi\in\conv(\Pi)$ such that $\umg(\pi)\subseteq G$, and (2) $\cw(G)=0$. When $\mg_\Pi\ne\emptyset$, we let $l_\Pi = \min_{G\in  \mg}\ties(G)$, where $\ties(G)$ denote the number of unordered pairs that are tied in $G$. 
  
\text{\bf Upper bound.} For any $n\in \mathbb N$ and any $\vec\pi\in \Pi^n$, we have:
$$
\Pr\nolimits_{P\sim\vec\pi}(\cw(P)=0)=\left\{\begin{array}{ll}O(n^{-\frac{l_\Pi}{2}})&\text{if  }\mg_\Pi\ne\emptyset\\ \exp(-\Omega(n))&\text{otherwise}\end{array}\right..$$

\text{\bf Tightness of the upper bound.} There exist infinitely many $n\in \mathbb N$ such that:
$$
\sup\nolimits_{\vec\pi\in\Pi^N}\Pr\nolimits_{P\sim\vec\pi}(\cw(P)=0)=\left\{\begin{array}{ll}\Omega(n^{-\frac{l_\Pi}{2}})&\text{if  }\mg_\Pi\ne\emptyset\\ \exp(-O(n))&\text{otherwise}\end{array}\right..
$$
\end{prop} 
\begin{proof} {\bf First step: defining $\bm\mc$}.   Let  $\mc=\{\pcons{G}: \cw(G) = 0\}$, where $\pcons{G}$ is defined in Definition~\ref{dfn:cg}. That is,  $\mc$ contains all $\pcons{G}$ where $G$ does not contain a Condorcet winner. We have the following observations. 
\begin{itemize}
\item [(1)] For any profile $P$, $\cw(P)=0$ if and only if $\hist(P)\in\bigcup_{\pcons{G}\in\mc}\ppoly{G}$. To see this, if  $\cw(P)=0$ then $\hist(P)\in\ppoly{\umg(P)}$, where $\pcons{\umg(P)}\in\mc$; and conversely, 
if $\hist(P)\in \ppoly{G}$ for some $\pcons{G}\in\mc$, then $P$ does not has a Condorcet winner. 
\item [(2)] For any $G$ with $\cw(G) = 0$,  we have $\ppoly{G}\ne\emptyset$ due to McGarvey's theorem~\cite{McGarvey53:Theorem}. 
\item [(3)] $|\mc|$ can be seen as a constant that does not depend on $n$. 
\end{itemize}

{\bf Second step: the polynomial case.}  To prove the upper bound, we note that for any $G$ such that (1) there is no Condorcet winner and (2) $G$ is a supergraph of the UMG of some $\pi\in\conv(\Pi)$, the number of ties in $G$ is at least $l_\Pi$, which means that $\rank(\pba{G})\ge l_\Pi$. Therefore, according to observation (1) above, we have: 
$$\Pr\nolimits_{P\sim\vec\pi}(\wcw(P)=k)\le \sum_{G: \pcons{G}\in \mc}\Pr\nolimits_{P\sim\vec\pi}(\hist(P)\in\ppoly{G})= O\left(n^{-\frac{l_\Pi}{2}}\right)$$ 
The last part follows after applying the polynomial upper bound in Lemma~\ref{lem:maintech} to all $\pcons{G}\in \mc$ and the observation (3) above. In particular, for any graph $G$ that is not a supergraph of the UMG of any $\pi\in\conv(\Pi)$, $\Pr_{P\sim\vec\pi}(\hist(P)\in\ppoly{G})$ is exponentially small due to Claim~\ref{claim:solumg} and the exponential upper bound in Lemma~\ref{lem:maintech} applied to $\pcons{G}$.

To prove the tightness, let $\pi\in \conv(\Pi)$ denote a distribution such that there exists a supergraph $G^*\in \mg_\Pi$ of $\umg(\pi)$ where $G^*$ contains $l_\Pi$ ties. By Claim~\ref{claim:solumg}, $\pi\in \ppolyz{G^*}$, which means that $\ppolyz{G^*}\cap  \conv(\Pi)\ne \emptyset$. Also by Claim~\ref{claim:solumg}, $\rank(\pba{G^*})=l_\Pi$. The tightness follows after applying the polynomial  tightness in Lemma~\ref{lem:maintech} to $\pcons{G^*}$.

{\bf Third step: the exponential case.}  For any $\pcons{G}\in \mc$ and any $\pi\in\conv(\Pi)$, we first prove that $\ppolyz{G}\cap \conv(\Pi)=\emptyset$. Suppose for the sake of contradiction that such $\pcons{G}\in\mc$ and $\pi\in {\ppolyz{G}}\cap \conv(\Pi)$ exist. It follows from Claim~\ref{claim:solumg} that $\umg(\pi)$ is a subgraph of $G$, which means that $G\in\mg_\Pi$. This contradicts the assumption that $\mg=\emptyset$. The upper bound (respectively, its tightness) follows after applying the exponential bound (respectively, its tightness) in Lemma~\ref{lem:maintech} to all  (respectively, an arbitrary) $\pcons{G}\in \mc$.
\end{proof}
The following claim implies that $l$ in Proposition~\ref{prop:nocw} can only be $0$ or $1$.
\begin{claim}\label{claim:cond01} For any unweighted directed graph $G$ over $\ma$ that does not contain a Condorcet winner, there exists a supergraph of $G$, denoted by $G^*$, such that $G^*$ does not contain a Condorcet winner, and the number of ties in $G^*$ is no more than one. The upper bound of one is tight. 
\end{claim}
\begin{proof}
Let $G^*$ denote a supergraph of $G$ without a Condorcet winner and with the minimum number of ties. If there is no tie in $G^*$ then the claim is proved. For any pair of tied alternatives $a$ and $b$ in $G^*$, adding $a\ra b$ to $G^*$ must lead to a Condorcet winner due to the minimality of $G^*$, and $a$ must be the Condorcet winner. This means that $a$ beats all alternatives other than $b$ in $G^*$. Similarly, $b$ beats all alternatives other than $a$ in $G^*$. This means that $\{a,b\}$  is the only tie in $G^*$ because if there exists another tie $\{c,d\}$, then adding $c\ra d$ to $G^*$ will not introduce a Condorcet winner, which contradicts the minimality of $G^*$. Therefore, the number of ties in $G^*$ is upper bounded by $1$. The tightness of the upper bound is proved by letting $G$ be a graph where $a$ and $b$ are the only weak Condorcet winners.\end{proof}

\begin{prop}[\bf Smoothed likelihood of exactly $k$ weak Condorcet winners]\label{prop:wcw} Let $\mm= (\Theta,\ml(\ma),\Pi)$ be a strictly positive and closed single-agent preference model.
  
\text{\bf Upper bound.} For any $n\in \mathbb N$, any $1\le k\le m$,  and any $\vec\pi\in \Pi^n$, we have:
$$
\Pr\nolimits_{P\sim\vec\pi}(\wcw(P)=k)=\left\{\begin{array}{ll}O(n^{-\frac{k(k-1)}{4}})&\text{if  } \exists\pi\in \conv(\Pi)\text{ s.t. }\wcw(\pi)\ge k\\ \exp(-\Omega(n))&\text{otherwise}\end{array}\right..$$

\text{\bf Tightness of the upper bound.} For any $1\le k\le m$, there exist infinitely many $n\in \mathbb N$ such that:
$$
\sup\nolimits_{\vec\pi\in\Pi^N}\Pr\nolimits_{P\sim\vec\pi}(\wcw(P)=k)=\left\{\begin{array}{ll}\Omega(n^{-\frac{k(k-1)}{4}})&\text{if  } \exists\pi\in \conv(\Pi)\text{ s.t. }\wcw(\pi)\ge k\\ \exp(-O(n))&\text{otherwise}\end{array}\right..
$$
\end{prop} 
\begin{proof} 
\noindent{\bf First step: defining $\bm\mc$}.   Let  $\mc=\{\pcons{G}: \wcw(G) = k\}$, where $\pcons{G}$ is defined in Definition~\ref{dfn:cg}. That is,  $\mc$ contains all $\pcons{G}$ where $G$ has exactly $k$ weak Condorcet winners. 
We have the following observations. 
\begin{itemize}
\item [(1)] For any profile $P$, $\wcw(P)=k$ if and only if $\hist(P)\in\bigcup_{\pcons{G}\in\mc}\ppoly{G}$. To see this, if $\wcw(P)=k$ then $\hist(P)\in\ppoly{\umg(P)}$, where $\umg(P)\in\mc$; and conversely, if $\hist(P)\in \ppoly{G}$, then $G=\umg(P)$ contains exactly $k$ weak Condorcet winners, which means that $\wcw(P)=k$. 
\item [(2)] by McGarvey's theorem~\cite{McGarvey53:Theorem}, for each $\pcons{G}\in\mc$, there exists a profile $P$ with $\umg(P)=G$, which means that $\ppoly{G}\ne\emptyset$. 
\item [(3)] The total number of unweighted directed graphs over  $\ma$ only depends on $m$, which means that $|\mc|$ can be seen as a constant that does not depend on $n$. 
\end{itemize}

\vspace{1mm}
\noindent{\bf Second step: the polynomial case.} To prove the upper bound, we note that for any $G$ that has exactly $k$ weak Condorcet winners, the number of ties is at least $\frac{k(k-1)}{2}$, which means that $\rank(\pba{G})\ge \frac{k(k-1)}{2}$. Therefore, according to observation (1) above, we have: 
$$\Pr\nolimits_{P\sim\vec\pi}(\wcw(P)=k)\le \sum_{G: \pcons{G}\in \mc}\Pr\nolimits_{P\sim\vec\pi}(\hist(P)\in\ppoly{G})= O\left(n^{-\frac{k(k-1)}{4}}\right)$$ 
The last part follows after applying the polynomial upper bound in Lemma~\ref{lem:maintech} to all $\pcons{G}\in \mc$ and the observation (3) above. 

To prove the tightness, suppose there exists $\pi\in\conv(\Pi)$ with $\wcw(\pi)\ge k$. This means that there exists a supergraph  of $\umg(\pi)$ over $\ma$, denoted by $G^*$, that has exactly $k$ weak Condorcet winners, and there is an edge from any weak Condorcet winner to any other alternative. By Claim~\ref{claim:solumg}, $\pi\in \ppolyz{G^*}$ and  $\rank(\pba{G^*})=\frac{k(k-1)}{2}$. The tightness follows after applying the tightness of the polynomial   bound in Lemma~\ref{lem:maintech} to $\pcons{G^*}$.

\vspace{1mm}
\noindent{\bf Third step: the exponential case.}  For any $\pcons{G}\in\mc$ and any $\pi\in\conv(\Pi)$, we first prove that $\ppolyz{G}\cap \conv(\Pi)=\emptyset$. Suppose for the sake of contradiction there exist $\pcons{G}\in\mc$ and $\pi\in {\ppolyz{G}}\cap \conv(\Pi)$. It follows that $\umg(\pi)$ is a subgraph of $G$, which means that all weak Condorcet winners in $G$ must also be weak Condorcet winners in $\pi$. Because there are $k$ weak Condorcet winners in $G$, we have $\wcw(\pi)\ge k$, which is a contradiction. The lower (respectively, upper) bound follows after applying Lemma~\ref{lem:maintech} to an arbitrary (respectively, all) $\pcons{G}\in \mc$.
\end{proof}
Consider the special case where  $\pi_\text{uni}\in \conv(\Pi)$. Notice that all edge weights in $\pi_\text{uni}$ are $0$, which means that $\pi_\text{uni}$ satisfies all pairwise constraints $\pair_{a,b}$ defined in Definition~\ref{dfn:varcons}. Consequently, for any $\ba$ and $\bb$ that only contain pairwise constraints, we have $\rsol\cap \conv(\Pi)\ne \emptyset$. This observation leads to the following corollary of Proposition~\ref{prop:wcw}.

\begin{coro}\label{coro:cc} Let $\mm= (\Theta,\ml(\ma),\Pi)$ be a  strictly positive and closed single-agent preference model with $\pi_\text{uni}\in \conv(\Pi)$.  For any $1\le k\le m$, any $n\in\mathbb N$, and any $\vec \pi\in\Pi^n$, we have $\Pr_{P\sim \vec\pi}(\wcw(P)=k)=O(n^{-\frac{k(k-1)}{4}})$. The bound is tight for infinitely many $n\in\mathbb N$ and corresponding $\vec \pi\in \Pi^n$. 
\end{coro}
When $\mm$ is neutral, we have $\pi_\text{uni}\in \conv(\Pi)$. This is because for any $\pi\in \Pi$, the average of $m!$ distributions obtained from $\pi$ by applying all permutations is $\pi_\text{uni}$. Therefore, Corollary~\ref{coro:cc} applies to all neutral,  strictly positive, and closed models including $\mm_\mallows^{[\underline{\varphi},1]}$ and $\mm_\pl^{[\underline{\varphi},1]}$ for all $0<\underline{\varphi}\le 1$.

\begin{prop}[\bf Smoothed likelihood of non-existence of weak Condorcet winners]\label{prop:nowcw} Let $\mm= (\Theta,\ml(\ma),\Pi)$ be a strictly positive and closed single-agent preference model.
  
\text{\bf Upper bound.} For any $n\in \mathbb N$ and any $\vec\pi\in \Pi^n$, we have:
$$
\Pr\nolimits_{P\sim\vec\pi}(\wcw(P)=0)=\left\{\begin{array}{ll}O(1)&\text{if  }\exists\pi\in \conv(\Pi)\text{ and }G\supseteq \umg(\pi)\text{ s.t. }\wcw(G)=0\\ \exp(-\Omega(n))&\text{otherwise}\end{array}\right.$$

\text{\bf Tightness of the upper bound.} There exist infinitely many $n\in \mathbb N$ such that:
$$
\sup\nolimits_{\vec\pi\in\Pi^N}\Pr\nolimits_{P\sim\vec\pi}(\wcw(P)=0)=\left\{\begin{array}{ll}\Omega(1)&\begin{array}{c}\text{if  }\exists\pi\in \conv(\Pi)\text{ and }G\supseteq \umg(\pi)\\ \text{ s.t. }\wcw(G)=0\end{array}\\ \exp(-O(n))&\text{otherwise}\end{array}\right.$$
\end{prop} 

\begin{proof}  {\bf First step: defining $\bm\mc$}.   Let  $\mc=\{\pcons{G}: \wcw(G) = 0\}$, where $\pcons{G}$ is defined in Definition~\ref{dfn:cg}. That is,  $\mc$ contains all $\pcons{G}$ where $G$ has no weak Condorcet winners. We have the following observations. 
\begin{itemize}
\item [(1)] For any profile $P$, $\wcw(P)=0$ if and only if $\hist(P)\in\bigcup_{\pcons{G}\in\mc}\ppoly{G}$. To see this, if $\wcw(P)=0$ then $\hist(P)\in\ppoly{\umg(P)}$, where $\umg(P)\in\mc$; and conversely, if $\hist(P)\in \ppoly{G}$, then $G=\umg(P)$ does not contain a weak Condorcet winner, which means that $\wcw(P)=0$. 
\item [(2)]  $\ppoly{G}\ne\emptyset$ due to McGarvey's theorem~\cite{McGarvey53:Theorem}. 
\item [(3)] $|\mc|$ can be seen as a constant that does not depend on $n$. 
\end{itemize}

{\bf Second step: the polynomial case.} The upper bound trivially holds. To prove the tightness, suppose there exit $\pi\in \conv(\Pi)$ and a supergraph $G$ of $\umg(\pi)$ that does not contain a Condorcet winner. Let $G^*$ denote an arbitrary complete supergraph of $G$. It follows that $G^*$ is a supergraph of $\umg(\pi)$ and $G^*$ does not contain a Condorcet winner, which means that $\pcons{G^*}\in\mc$. By Claim~\ref{claim:solumg}, $\pi\in \ppolyz{G^*}$, which means that $\ppolyz{G^*}\cap \conv(\Pi) \ne\emptyset$. Note that $\pba{G^*}=\emptyset$, which means that $\rank(\pba{G^*})=0$. The tightness follows after applying the tightness of the polynomial   bound in Lemma~\ref{lem:maintech} to $\pcons{G^*}$.

{\bf Third step: the exponential case.}  For any $\pcons{G}\in\mc$ and any $\pi\in\conv(\Pi)$, we first prove that $\ppolyz{G}\cap \conv(\Pi)=\emptyset$. Suppose for the sake of contradiction that there exist $\pcons{G}\in\mc$ and  $\pi\in {\ppolyz{G}}\cap \conv(\Pi)$, which means that $\wcw(G)=0$. By Claim~\ref{claim:solumg}, $\umg(\pi)$ is a subgraph of $G$, which contradicts the assumption of the exponential case, that all supergraphs of $\umg(\pi)$  contains at least one weak Condorcet winner. The upper bound (respectively, its tightness) follows after applying the exponential bound  (respectively, its tightness) in Lemma~\ref{lem:maintech} to all (respectively, an arbitrary) $\pcons{G}\in \mc$.
\end{proof}


\end{document}